\documentclass[11pt,a4paper]{article}

\usepackage{makecell}
\usepackage{authblk}
\usepackage{amsmath,amssymb,amsthm}
\allowdisplaybreaks
\usepackage[margin=1in]{geometry}
\usepackage{color}
\usepackage{float}
\usepackage{enumerate}
\usepackage{listings}
\usepackage{xcolor}
\usepackage[title]{appendix}
\newtheorem{theorem}{Theorem}[section]
\newtheorem{lemma}[theorem]{Lemma}
\newtheorem{proposition}[theorem]{Proposition}
\newtheorem{corollary}[theorem]{Corollay}

\newtheorem{remark}[theorem]{Remark}
\newtheorem{conjecture}[theorem]{Conjecture}

\numberwithin{equation}{section}

\begin{document}

\title{Recurrence Coefficients of Oscillatory Gegenbauer-type\\ Orthogonal Polynomials}

\author[1,*]{Shulin Lyu}
\author[1,$\dag$]{Xun Zhou}  
\affil[1]{School of Mathematics and Statistics, Qilu University of Technology (Shandong Academy of Sciences), Jinan 250353, China}
\affil[*]{\texttt{lvshulin1989@163.com}, Corresponding author}
\affil[$\dag$]{\texttt{zhouxun200212@163.com}}
\date{\today}
\maketitle

\begin{abstract}
We study two classes of oscillatory Gegenbauer-type weight functions. One is $w_1(x)=(1-x^2)^{\lambda-1/2}\exp(i\zeta x),x\in[-1,1]$, where $\lambda>-1/2$ is rational and $\zeta\in\mathbb{R}\setminus\{0\}$ is a zero of the Bessel function $J_{\lambda-1}$. The other is $w_2(x)=x(1-x^{2})^{\gamma-1/2} \exp(\mathit{i}\eta x), x\in[-1,1]$, where $\gamma$ and $\eta$ are chosen such that the corresponding orthogonal polynomials exist. By using the ladder operators established in recent work for monic orthogonal polynomials associated with Jacobi-type weight functions and three compatibility conditions, we derive two coupled difference equations for the three-term recurrence coefficients corresponding to $w_1$ and $w_2$, respectively. For $w_1$, and for $w_2$ where $\gamma=0$ and $\eta$ is a positive zero of $J_0$, or $\gamma=1/2$ and $\eta=m\pi$ with $m$ denoting a nonzero integer, we determine the initial values for the recurrence coefficients and compute the first few terms via the difference equations. Then we conjecture general-$n$ symbolic forms for the recurrence coefficients and for $\textbf{p}(n)$, the coefficient of $x^{n-1}$ in the monic orthogonal polynomial of degree $n$. The case $\lambda=0$ in $w_1$ is treated separately, since the conjectured general-$\lambda$ forms degenerate here, yielding substantially different results. For $w_1$, we also derive a third-order difference equation satisfied by $\textbf{p}(n)$.
\end{abstract}

\section{Introduction}
We consider two classes of oscillatory Gegenbauer-type weight functions. The first class is
\begin{align}\label{w1}
w_1(x)=\left(1-x^2\right)^{\lambda-1/2}\exp(i\zeta x), \qquad x\in[-1,1],
\end{align}
where $\lambda$ is a rational number, $\lambda>-1/2$, and $\zeta\in\mathbb{R}\setminus\{0\}$ is a zero of the Bessel function of the first kind $J_{\lambda-1}$. Here $i$ is the imaginary unit, i.e., $i^2=-1$, and $J_{\nu}$ is defined by \cite[p. 102]{Lebedev}:
\[J_{\nu}(z)=\sum_{k=0}^{\infty}\frac{(-1)^k(z/2)^{\nu+2k}}{k!\,\Gamma(\nu+k+1)}.\]
The second class is
\begin{align}\label{w2}
w_2(x) = x\left(1-x^{2}\right)^{\gamma-1/2} \exp(\mathit{i}\eta x),	 \qquad x\in[-1,1] ,\quad \gamma >-1/2.
\end{align}
The real parameters $\gamma$ and $\eta$ are chosen such that the corresponding orthogonal polynomials (OPs for short) exist.

It was demonstrated in \cite[Theorem 2.2]{MCM} that there exists a sequence of monic polynomials orthogonal with respect to $w_1(x)$, which extends the results of \cite{MCS} where $\lambda$ is restricted to be a positive rational number. In addition, the symbolic forms of the three-term recurrence coefficients for the monic OPs were conjectured. Moreover, difference equations were deduced for the recurrence coefficients, using the differential relations for the OPs. The existence of the OPs associated with the weight function $(1-x)^{\alpha-1/2}(1+x)^{\beta-1/2}\exp(i\zeta x),x\in[-1,1]$, was established in \cite{SC}, where $\alpha,\beta>-1/2$ are rational numbers, $\beta=\alpha\pm \ell$ with $\ell$ being a positive integer, and $\zeta$ is a positive zero of $J_{\alpha-1}$. Note that the special case $\alpha=\beta=\lambda$, i.e., $\ell=0$, which reduces to \eqref{w1}, was excluded. Difference equations satisfied by the recurrence coefficients, though not presented explicitly, are obtainable from the given results in \cite{SC} by one further step (see Theorems 3.3-3.5 therein). Given initial values, the recurrence coefficients can be computed  from these equations, which enables the authors of \cite{SC} to construct the effective quadrature rules of Gaussian type. This weight function from \cite{SC} was also considered in \cite{ZhouLyu26}. By employing two conditions compatible to the ladder operators satisfied by the associated monic OPs, a pair of coupled diffference equations were deduced for the recurrence coefficients, with the initial values related to the Kummer function. The weight function $(1-x)^{\alpha}(1+x)^{\beta}\exp(-tx), x\in[-1,1]$ with $\alpha,\beta>0, t\in\mathbb{R}$ was studied in \cite{BCE}. By using the ladder operators and  compatibility conditions, two coupled difference equations were established for the two introduced auxiliary quantities. Additionally, by differentiating the orthogonality relations, the associated Hankel determinant was expressed as an integral in terms of the solution to the  $\sigma$-form of a Painlev\'{e} V equation. The above-mentioned results of \cite{BCE} were generalized to $\alpha,\beta>-1$ in \cite{LyuLyu25}.

For $w_2(x)$ given by \eqref{w2}, the existence of monic OPs was proved when $\gamma=0$ and $\eta$ is a positive zero of $J_0$ (see \cite[Theorem 4]{MC05-1}), and when $\gamma=1/2$ and $\eta=m\pi$ with $m$ denoting a nonzero integer (see \cite[Theorem 2.3]{MC05-2}). For the former case (i.e. $w_2(x)=x(1-x^2)^{-1/2}\exp(i\eta x)$), as the degree of the monic OP tends to $\infty$, the asymptotic formulas for the recurrence coefficients were obtained, based on the asymptotic properties of the OPs associated with the weight $(1-x^2)^{-1/2}\exp(i\eta x)$. For the recurrence coefficients of the latter case (i.e. $w_2(x)=x\exp(i\eta x)$), symbolic forms were conjectured with the aid of a software package and difference equations were deduced by using differential identities satisfied by the OPs. Gaussian quadrature rules were also developed and applied to the numerical evaluation of highly oscillatory integrals.

Oscillatory OPs and computations of oscillatory integrals by using quadrature rules were investigated in several studies. The existence of the monic OPs associated with $x^{\alpha}{\rm e}^{-x}\exp(i\zeta x), x\in[0,+\infty),\alpha>-1,\zeta>0$, was proved in \cite{MSM} and the corresponding quadrature rules of Gaussian type was constructed. The highly oscillatory integrals with respect to $[\log(x-a)]^{c_1}[\log(b-x)]^{c_2}\exp(i\omega x)(x-a)^{-\alpha}(b-x)^{-\beta}, x\in[a,b], c_1,c_2\in\{0,1\},\omega\gg1, \alpha,\beta<1$, was evaluated numerically in \cite{KHM} by using the recurrence coefficients of the OPs associated with Gautschi's weight $t^{s}{\rm e}^{-t}(t-1-\log t), t\in(0,\infty),s>-1$ and the corresponding Gaussian-type quadrature rules. The oscillatory integrals over $x\in[0,b]$ with the kernel $x^{\alpha}(b-x)^{\beta}\exp(i\omega x^{\gamma})$, where $-1<\alpha,\beta\leq0, |\omega|\gg1$, and $ \gamma$ is a positive integer, was computed in \cite{CYC} via an $s$-step asymptotic rule and asymptotic-Filon-type quadrature rule. See also \cite{KX} for analysis on four kinds of highly oscillatory Fourier-type integrals.

Denote $w_1(x)$ and $w_2(x)$ by $w(x)$, and the associated monic OPs by $\{P_n(x)\}$. Write
\begin{align}\label{mp}
P_n(x):=x^n+\textbf{p}(n)x^{n-1}+\dots+P_n(0),\quad\qquad n=1,2,\ldots,
\end{align}
with $P_0(x):=1$ and $\textbf{p}(0):=0$, and suppose
\begin{align}\label{dp}	
\int_{-1}^1P_m(x)P_n(x)w(x)dx=h_n\delta_{mn},\quad\qquad m,n=0,1,\ldots,
\end{align}
where $\delta_{mn}$ is the Kronecker delta function, i.e., $\delta_{mn}=1$ for $m=n$ and $0$ otherwise. According to \eqref{mp}-\eqref{dp}, ones gets the three-term recurrence relation
\begin{align}\label{3r}
x\,P_n(x)=P_{n+1}(x)+i\alpha_nP_n(x)+\beta_nP_{n-1}(x),\qquad\qquad n\geq0,
\end{align}
with $P_{-1}(x):=0$. In much of the literature, $\beta_0$ is set to 0, but in this paper we keep it arbitrary. The recurrence coefficients  are given by
\begin{align}
i\alpha_n=&\frac{1}{h_n}\int_{-1}^1xP_n^2(x)w(x)dx\label{al-1}\\
=&\textbf{p}(n)-\textbf{p}(n+1),\qquad\qquad n\geq0,\label{al}\\
\beta_n=&\frac{1}{h_{n-1}}\int_{-1}^1xP_n(x)P_{n-1}(x)w(x)dx\label{bt-1}\\
=&\frac{h_n}{h_{n-1}},\qquad\qquad n\geq1.\label{bt}
\end{align}
Refer to \cite[p. 22]{Ismail} and \cite{ML25}. It follows from \eqref{al} that
\begin{align}
\mathit{i}\sum\limits_{j=0}^{n-1} \alpha_{j}=-\textbf{p}(n).\label{(1.7)}
\end{align}

In \cite{LyuLyu25}, the ladder operators and three compatibility conditions were established for the monic OPs associated with the Jacobi-type weight function
$w_{_J}(x)=(1-x)^{\alpha}(1+x)^{\beta}w_0(x), x\in[-1,1],~\alpha,\beta>-1$. Note that two assumptions were made on $w_{_J}(x)$: (i) the moments of $w_{_J}(x)$, i.e. $\int_{-1}^1 x^jw_{_J}(x)dx, j=0,1,\ldots$, are finite; (ii) $w_0(x)>0$ is continuously differentiable and bounded at $\pm1$.
The weight functions given by \eqref{w1} and \eqref{w2} correspond to $w_0(x)=\exp(i\zeta x)$ and $w_0(x)=x\exp(\mathit{i}\eta x)$ respectively which satisfy Assumption (ii) except for the condition $w_0(x)>0$. However, this condition and Assumption (i) are imposed to guarantee the existence of OPs, and it was not used in the derivation of the ladder operators and three compatibility conditions in \cite{LyuLyu25}. Therefore, these results also apply to the weight functions given by \eqref{w1} and \eqref{w2} whose associated OPs are known to exist. We restate them below.
\begin{proposition}
The monic orthogonal polynomials $\{P_n\}$ defined by \eqref{mp}-\eqref{dp} satisfy the following ladder operators:
\begin{align}
P_n'(z)=&-B_n(z)P_n(z)+\beta_nA_n(z)P_{n-1}(z),\label{loJ}\\
P_{n-1}'(z)=&\left(B_n(z)+v'(z)\right)P_{n-1}(z)-A_{n-1}(z)P_n(z),\label{roJ}
\end{align}
for $n\geq0$, where $v(x):=-\ln w(x)$. The coefficient functions $A_n(z)$ and $B_n(z)$ are defined by
\begin{align}\label{An}
A_n(z):=\frac{1}{1-z^2}\left[  \frac{1}{h_n}\int_{-1}^{1}\frac{(1-z^2)v'(z)-(1-x^2)v'(x)}{z-x}P_n^2(x)w(x)dx+2n+1 \right],
\end{align}
for $n\geq0$ with $A_{-1}(z):=0$, and
\begin{align}\label{Bn}		 B_n(z):=&\frac{1}{1-z^2}\cdot\frac{1}{h_{n-1}}\int_{-1}^{1}\frac{(1-z^2)v'(z)-(1-x^2)v'(x)}{z-x}P_n(x)P_{n-1}(x)w(x)dx+\frac{nz-\textbf{p}(n)}{1-z^2}, \end{align}
for $n\geq1$ with $B_0(z):=0$. Furthermore, $A_n(z)$ and $B_n(z)$ satisfy the following three compatibility conditions
 \begin{align}
B_{n+1}(z)+B_n(z)&=(z-i\alpha_n)A_n(z)-v'(z)\tag{$S_1$},\\
1+(z-i\alpha_n)(B_{n+1}(z)-B_n(z))&=\beta_{n+1}A_{n+1}(z)-\beta_nA_{n-1}(z)\tag{$S_2$},			 \end{align}
for $n\geq0$, and
\begin{align}
B^2_n(z)+v'(z)B_n(z)+\sum_{j=0}^{n-1}A_j(z)&=\beta_nA_n(z)A_{n-1}(z)\tag{$S^{'}_2$},
\end{align}
for $n\geq1$.
\end{proposition}

It should be noted that, in previous work, the parameters $\alpha$ and $\beta$ in $w_{_J}(x)$ were usually assumed to be positive, and the expressions for $A_n$ and $B_n$ are different from \eqref{An}-\eqref{Bn}. These studies made use of the ladder operators and $\{(S_1), (S_2), (S_2')\}$ to investigate the monic OPs, the partition functions and gap probabilities of the unitary ensembles, etc. See, for instance, \cite{BCH, CI, MC17, MC20, MC21}. With regard to the Laguerre-type weight functions $x^{\lambda}w_0(x), x\in[0,+\infty)$, see, e.g., \cite{CI10, CM12, ML25} for the applications of the ladder operators to problems with $\lambda>0$ and \cite{LyuLyu25, MF25} for the derivation of the ladder operators with $\lambda>-1$. Using the difference and Toda-type equations established for the recurrence coefficients via the ladder operators, researchers deduced large-$n$ and large-time asymptotics for the recurrence coefficients and related quantities. We refer to recent work, e.g., \cite{CJ21,MF25, MW26-1, MW26-2}.

For $w_1(x)$ and $w_2(x)$,  by substituting the expressions of $A_n$ and $B_n$ obtained from \eqref{An}-\eqref{Bn} into $\{(S_1), (S_2),(S_2')\}$, we deduce two coupled difference equations satisfied by $\alpha_n$ and $\beta_n$, which can be iterated in $n$ given initial values. Specifically, Section 2 focuses on $w_1$. A third-order difference equation is also derived for $\textbf{p}(n)$.
After the first few terms of $\{\alpha_n,\beta_n,\textbf{p}(n)\}$ are computed, their symbolic forms for general $n$ are conjectured. The degenerate case $\lambda=0$ is addressed separately, since it differs substantially from the general case. When $\lambda=1$ and $\zeta$ is a particular zero of $J_0$, $\alpha_n$ and $\beta_n$ are computed numerically via the difference equations they satisfy. Section 3 deals with $w_2$ and two auxiliary quantities are introduced. For two special cases of $(\gamma,\eta)$: (i) $\gamma=0$ and $\eta$ is a positive zero of $J_0$; (ii) $\gamma=1/2$ and $\eta=m\pi$ with $m$ denoting a nonzero integer, we conjecture symbolic forms of $\alpha_n,\beta_n$, and $\textbf{p}(n)$ for arbitrary $n$. When $m=100$ and $m=1$ in case (ii), $\alpha_n$ and $\beta_n$ are calculated numerically. Section 4 presents our conclusions.
Detailed computations and the numerical implementation in Mathematica are provided in the Appendices.

\section{Oscillatory Gegenbauer Weight and Difference Equations}
For the weight function given by \eqref{w1}, we have
\[v(x)=-\ln w_1(x)=-i\zeta x-\left(\lambda-\frac{1}{2}\right)\ln (1-x^2),\]
and hence
\[v'(x)=-i\zeta+\frac{(2\lambda-1)x}{1-x^2}.\]
It follows that
\[\frac{(1-z^2)v'(z)-(1-x^2)v'(x)}{z-x}=i\zeta x+i\zeta z+2\lambda-1.\]
Plugging it into \eqref{An}-\eqref{Bn}, in view of the orthogonality relation \eqref{dp} and the integral formulas for  $\{\alpha_n,\beta_n\}$ given by \eqref{al-1} and \eqref{bt-1}, we come to the expressions for $A_n$ and $B_n$.
\begin{lemma}
$A_n(z)$ and $B_n(z)$ are given by
\begin{align}
A_n(z)=&\frac{i\zeta z-\zeta\alpha_n+2(n+\lambda)}{1-z^2},\qquad\qquad n\geq0,\label{An-J}\\
B_n(z)=&\frac{nz+i\zeta\beta_n-\textbf{p}(n)}{1-z^2},\qquad\qquad n\geq1.\label{Bn-J}
\end{align}
\end{lemma}

Substituting \eqref{An-J}-\eqref{Bn-J}, together with the initial values $A_{-1}(z)=0$ and $B_0(z)=0$, into $(S_1)$ and $(S_2')$, we obtain the following theorem.
\begin{theorem}
$\alpha_n$ and $\beta_n$ satisfy the following two coupled difference equations
\begin{align}
\beta_{n+1}=&\beta_{n}\left(\frac{2\zeta\left(\alpha_{n}+\alpha_{n-1}\right)}{2n+2\lambda-1}-3\right)+\alpha_{n}\left(\alpha_{n}-\frac{2n+2\lambda+1}{\zeta}\right)+1,\label{diff1}\\
\alpha_{n+1}=&-\alpha_{n}+\frac{2n+2\lambda+1 }{2\zeta}\left[3+\frac{\beta_{n}-1}{\beta_{n+1}}+\frac{\alpha_{n}}{\beta_{n+1}}\left(\frac{2n+2\lambda-1}{\zeta}-\alpha_{n}\right)\right],\label{diff2}
\end{align}
for $n\geq1$, which can be iterated in $n$ given the initial conditions:
\begin{align}\label{al0bt1}
\alpha_0=\frac{2\lambda}{\zeta},\quad\qquad\alpha_1=\frac{2(\lambda+1)}{\zeta}-\frac{\zeta(2\lambda+1)}{\zeta^2-2\lambda},\quad\qquad
\beta_{1}=1-\frac{2\lambda}{\zeta^2}.
\end{align}
Note that $\beta_0$  was specified as $\sqrt{\pi}\,\Gamma\left(\lambda+1/2\right)\left(2/\zeta\right)^{\lambda}J_{\lambda}(\zeta)$ in \cite{MCM}, whereas in much of the literature it is taken to be zero. We leave it arbitrary in this paper.
\end{theorem}
\begin{proof}
Inserting \eqref{An-J}-\eqref{Bn-J} into $(S_1)$, and equating the coefficients of $(z-1)^{-1}$ and $(z+1)^{-1}$ on both sides, we obtain
\begin{align}\label{S1}
i\zeta\left(\beta_{n+1}+\beta_n\right)-\textbf{p}(n+1)-\textbf{p}(n)=i\alpha_n\left(\zeta\alpha_n-2n-2\lambda\right)+i\zeta,\qquad n\geq1.
\end{align}
The restriction $n\geq1$ rather than  $n\geq0$ is due to the fact that the expression for $B_n$ in \eqref{Bn-J} is valid for $n\geq1$. According to \eqref{al}, we replace $\textbf{p}(n+1)$ with $\textbf{p}(n)-i\alpha_n$ in \eqref{S1} and get
\begin{align}\label{S1-1}
\textbf{p}(n)=\frac{1}{2}\left[i\zeta(\beta_{n+1}+\beta_n-1)+i\alpha_n(-\zeta\alpha_n+2n+2\lambda+1)\right],\qquad n\geq1.
\end{align}
Replacing $n$ by $n+1$ in this equality and subtracting the resulting equation from \eqref{S1-1}, in light of the fact that $\textbf{p}(n)-\textbf{p}(n+1)=i\alpha_n$ given by \eqref{al}, we are led to
\begin{align}\label{w1diff1-1}
\beta_{n+2}=\beta_n+\alpha_{n+1}^2-\alpha_n^2+\frac{1}{\zeta}\big[\alpha_n(2n+2\lambda-1)-\alpha_{n+1}(2n+2\lambda+3)\big], \qquad n\geq1.
\end{align}
To continue, we plug \eqref{An-J}-\eqref{Bn-J} into $(S_2')$, multiply both sides of the obtained identity by $\left(1-z^2\right)^2$ and find
\begin{equation*}
\begin{aligned}
&(2n+2\lambda-1)\left(i\zeta\beta_n-\textbf{p}(n)\right)\cdot z+(i\zeta\beta_n-\textbf{p}(n))^2+n(n+2\lambda-1)\\
&=i\zeta\beta_n\big(-\zeta(\alpha_n+\alpha_{n-1})+2(2n+2\lambda-1)\big)\cdot z\\
&\quad+\beta_n\big(-\zeta\alpha_n+2n+2\lambda\big)\big(-\zeta\alpha_{n-1}+2n+2\lambda-2\big)-\zeta^2\beta_n,
\end{aligned}
\end{equation*}
for $n\geq1$. Comparing the coefficient of $z$ on both sides yields
\begin{align}\label{S2'-1}
(2n+2\lambda-1)\left(i\zeta\beta_n+\textbf{p}(n)\right)=&i\zeta^2\beta_n\left(\alpha_n+\alpha_{n-1}\right),\qquad n\geq1.
\end{align}
Substituting \eqref{S1-1} into \eqref{S2'-1} gives us \eqref{diff1}.
A combination of it with $n$ replaced by $n+1$ and \eqref{w1diff1-1} leads us to \eqref{diff2}.

The detailed computations of $\{\alpha_0,\alpha_1,\beta_1\}$ are presented in Appendix \ref{J0D-int-1}. Alternatively, their expressions can be derived using the Mathematica code provided in Appendix \ref{IV1-M}.
\end{proof}

\begin{remark}
The system of coupled difference equations \eqref{diff1}-\eqref{diff2}, as a whole, is second order in $\alpha_n$ and first order in $\beta_n$, while each equation is first order in both variables.
A coupled pair of difference equations were also deduced for $\alpha_n$ and $\beta_n$ in \cite[Theorem 3.5]{MCM} by using some differential relations satisfied by the monic orthogonal polynomials. One equation is the same as \eqref{w1diff1-1} and the other is second order in both $\alpha_n$ and $\beta_n$. To iterate this system, the initial conditions  $\{\alpha_1,\alpha_2,\beta_1,\beta_2\}$ are needed.
\end{remark}

\begin{remark}
The weight function $(1-x)^{\alpha}(1+x)^{\beta}{\rm e}^{-tx}, x\in[-1,1]$ with $t\in\mathbb{R}$ was studied in \cite{BCE} for $\alpha,\beta>0$ and generalized to $\alpha,\beta>-1$ in \cite{LyuLyu25}. By using the ladder operators satisfied by the associated monic orthogonal polynomials and $\{(S_1), (S_2')\}$, two coupled difference equations were established for the two introduced auxiliary quantities. However, by replacing $t$ with $-i\zeta$ and $\{\alpha,\beta\}$ with $\lambda-1/2$ in their difference equations, it is not straightforward to derive \eqref{diff1}-\eqref{diff2}.

In addition, the weight function $(1-x)^{\alpha-1/2}(1+x)^{\beta-1/2}\exp(i\zeta x),x\in[-1,1]$, was considered in \cite{ZhouLyu26}, where $\alpha,\beta>-1/2$ are rational numbers, $\beta=\alpha\pm \ell$ with $\ell$ being a positive integer, and $\zeta$ is a positive zero of $J_{\alpha-1}$. By employing $(S_1)$ and $(S_2)$, a pair of coupled diffference equations were deduced for the recurrence coefficients, with the initial values related to the Kummer function. Note that although the special case $\alpha=\beta=\lambda$ {\rm(}i.e. $\ell=0${\rm)}, which reduces to \eqref{w1}, was excluded in \cite{ZhouLyu26}, the difference equations (22)-(23) therein still hold, which as a whole is one order higher in $\beta_n$ than \eqref{diff1}-\eqref{diff2}.
\end{remark}

\begin{corollary}
When $\lambda\geq0$, $\zeta$ {\rm(}a zero of $J_{\lambda-1}${\rm)} is real \cite[p. 482]{Watson}, and so are the initial values given by \eqref{al0bt1}.
Since the coefficients that appear in \eqref{diff1}-\eqref{diff2} are all real, we conclude that the sequences $\{\alpha_n\}_{n\geq0}$ and $\{\beta_n\}_{ n\geq1}$ are real. This coincides with \cite[Lemma 3.1]{MCS} which was derived using the properties of the moments of $w_1$. In addition, according to  \eqref{(1.7)}, i.e. $\textbf{p}(n)=-\mathit{i}\sum_{j=0}^{n-1} \alpha_{j}$, we know that $\textbf{p}(n)$ is purely imaginary.
\end{corollary}

With the initial values given in \eqref{al0bt1}, we successively compute $\{\alpha_j,\beta_j\}_{2\leq j\leq9}$ by iterating \eqref{diff1}-\eqref{diff2} for $n=1,\ldots,8$. Given $\textbf{p}(0)=0$, by using the relation $\textbf{p}(n+1)=\textbf{p}(n)-i\alpha_n$ for $n=0,1,\ldots,8$, we obtain $\{\textbf{p}(j)\}_{1\leq j\leq9}$.  The computed results, together with the initial values, are listed in Table \ref{Table-1}. Note that our expressions of $\{\alpha_0,\alpha_1,\alpha_2, \beta_1,\beta_2\}$ are identical to those in \cite[Table 1]{MCM}, which were obtained using a software package.
\begin{table}[htbp]
\centering
\renewcommand{\arraystretch}{1.36}
\setlength{\tabcolsep}{2pt}
\caption{$\{\alpha_n,\beta_n,\textbf{p}(n)\}_{0\leq n\leq9}$ with general $\lambda$ ($\lambda\in\mathbb{Q},\lambda>-1/2)$ and arbitrary $\beta_0$}\label{Table-1}
\begin{tabular}{|c|c|c|c|}
\hline
$n$ & $\alpha_n$ & $\beta_n$ & $\textbf{p}(n)$  \\
\hline
0 & $\frac{2\lambda}{\zeta}$ & $\beta_{0}$ & $0$ \\
\hline
1 & $\frac{2(\lambda+1)}{\zeta}-\frac{\zeta(2\lambda+1)}{\zeta^2-2\lambda}$ & $\frac{\zeta^2-2\lambda}{\zeta^2}$ & $-\mathit{i}\cdot\frac{2\lambda}{\zeta}$ \\
\hline
2 & $\frac{2(\lambda+2)}{\zeta}+\zeta\left(\frac{2\lambda+1}{\zeta^2-2\lambda}-\frac{(2\lambda+3)(\zeta^2-4\lambda^2-4\lambda)}{2X_{2}}\right)$ & $-\frac{2(2\lambda+1)X_{2}}{\zeta^2(\zeta^2-2\lambda)^2}$ & $-\mathit{i}(2\lambda+1)\left(\frac{2}{\zeta}-\frac{\zeta}{\zeta^2-2\lambda}\right)$ \\
\hline
3 & $\frac{2(\lambda+3)}{\zeta}+(2\lambda+3)\zeta\left(\frac{\zeta^2-4\lambda^2-4\lambda}{2X_{2}}-\frac{2X_{3}}{X_{4}}\right)$ & $\frac{(\zeta^2-2\lambda)X_{4}}{\zeta^2X_{2}^2}$ & $-\mathit{i}\left(\frac{6(\lambda+1)}{\zeta}-\frac{(2\lambda+3)\zeta(\zeta^2-4\lambda^2-4\lambda)}{2X_{2}}\right)$ \\
\hline
4 &
$\frac{2(\lambda+4)}{\zeta}+\zeta\left(\frac{2(2\lambda+3)X_{3}}{X_{4}}-\frac{(2\lambda+5)X_{5}}{X_{6}}\right)$ & $-\frac{4(2\lambda+3)X_{2}X_{6}}{\zeta^2X_{4}^2}$ & $-2\mathit{i}(2\lambda+3)\left(\frac{2}{\zeta}-\frac{\zeta X_{3}}{X_{4}}\right)$ \\
\hline
5 & $\frac{2(\lambda+5)}{\zeta}+(2\lambda+5)\zeta\left(\frac{X_{5}}{X_{6}}-\frac{3X_{8}}{X_{9}}\right)$ & $\frac{X_{4}X_{9}}{\zeta^2X_{6}^2}$ & $-\mathit{i}\left(\frac{10(\lambda+2)}{\zeta}-\frac{(2\lambda+5)\zeta X_{5}}{X_{6}}\right)$ \\
\hline
6 & $\frac{2(\lambda+6)}{\zeta}+3\zeta\left(\frac{(2\lambda+5)X_{8}}{X_{9}}-\frac{(2\lambda+7)X_{11}}{2X_{12}}\right)$ & $-\frac{6(2\lambda+5)X_{6}X_{12}}{\zeta^2X_{9}^2}$ & $-3\mathit{i}(2\lambda+5)\left(\frac{2}{\zeta}-\frac{\zeta X_{8}}{X_{9}}\right)$ \\
\hline
7 & $\frac{2(\lambda+7)}{\zeta}+(2\lambda+7)\zeta\left(\frac{3X_{11}}{2X_{12}}-\frac{4X_{15}}{X_{16}}\right)$ & $\frac{X_{9}X_{16}}{\zeta^2X_{12}^2}$ & $-\mathit{i}\left(\frac{14(\lambda+3)}{\zeta}-\frac{3(2\lambda+7)\zeta X_{11}}{2X_{12}}\right)$ \\
\hline
8 & $\frac{2(\lambda+8)}{\zeta}+2\zeta\left(\frac{2(2\lambda+7)X_{15}}{X_{16}}-\frac{(2\lambda+9){X_{19}}}{X_{20}}\right)$ & $-\frac{8(2\lambda+7)X_{12}X_{20}}{\zeta^2X_{16}^2}$ & $-4\mathit{i}(2\lambda+7)\left(\frac{2}{\zeta}-\frac{\zeta X_{15}}{X_{16}}\right)$ \\
\hline
9 & $-$ & $\frac{X_{16}X_{25}}{\zeta^2X_{20}^2}$ & $-2\mathit{i}\left(\frac{9(\lambda+4)}{\zeta}-\frac{(2\lambda+9)\zeta X_{19}}{X_{20}}\right)$ \\
\hline
\end{tabular}
\parbox{0.95\textwidth}{\footnotesize \textbf{Note.} $X_j$ stands for $X_j(\zeta^2)$, with $X_j(\cdot)$ denoting a monic polynomial of degree $j$. Hence, $X_j(\zeta^2)$ is a monic polynomial of degree $2j$ in $\zeta$, containing only even powers of $\zeta$. The expressions of $X_j(\zeta^2)$ for $j=2,3,4,5,6,8,9,11,12$ are given in Appendix \ref{symbols},
while those for $j=15,16,19,20,25$ are not presented for brevity. We do not present $\alpha_9$ here since its computation took an excessively long time. This is possibly because its form is complex and reducing it to the desired form is rather involved.}
\end{table}

To proceed, with the aid of \eqref{al}, we use \eqref{S1} and \eqref{S2'-1} to establish the difference equation for $\textbf{p}(n)$.
\begin{theorem}\label{THdiffpn}
$\textbf{p}(n)$ satisfies the following third-order difference equation
\begin{equation}\label{diffpn}
\begin{aligned}
&\mathit{i}\zeta \left(\textbf{p}(n)-\textbf{p}(n+1)\right)^2+2\lambda_n(\textbf{p}(n)-\textbf{p}(n+1))-\textbf{p}(n)-\textbf{p}(n+1)-i \zeta\\
&\qquad=\frac{(2\lambda_n+1)\textbf{p}(n+1)}{2\lambda_n+1+\mathit{i}\zeta \left(\textbf{p}(n)-\textbf{p}(n+2)\right)}+\frac{(2\lambda_n-1)\textbf{p}(n)}{2\lambda_n-1+\mathit{i}\zeta \left(\textbf{p}(n-1)-\textbf{p}(n+1)\right)},
\end{aligned}
\end{equation}
for $n\geq1$ with $\lambda_n:=n+\lambda$,  which can be solved recursively given the initial values
\[\textbf{p}(0)=0,\quad\qquad \textbf{p}(1)=-\frac{2i\lambda}{\zeta},\quad\qquad \textbf{p}(2)=-i(2\lambda+1)\left(\frac{2}{\zeta}-\frac{\zeta}{\zeta^2-2\lambda}\right).\]
\end{theorem}

\begin{proof}
Eliminating $\alpha_n$ and $\alpha_{n-1}$ from \eqref{S2'-1} via \eqref{al}, we get
\begin{align}\label{btp}
\beta_{n}=\frac{(2\lambda_n-1)\mathit{i}\textbf{p}(n)}{\zeta\left(2\lambda_n-1+\mathit{i}\zeta \left(\textbf{p}(n-1)-\textbf{p}(n+1)\right)\right)},\qquad n\geq1.
\end{align}	
Using it to replace $\beta_{n+1}$ and $\beta_{n}$ in \eqref{S1}, and applying \eqref{al} to remove $\alpha_{n}$, we come to \eqref{diffpn}.

Noting that $\textbf{p}(n+1)=\textbf{p}(n)-i\alpha_n$ for $n\geq0$ and $\textbf{p}(0)=0$, with the values of $\alpha_0$ and $\alpha_1$ given in \eqref{al0bt1}, we obtain $\textbf{p}(1)$ and $\textbf{p}(2)$.
\end{proof}

\begin{remark}
The expressions of $\{\alpha_n,\beta_n,\textbf{p}(n)\}$ presented in Table \ref{Table-1} can also be derived using \eqref{diffpn}-\eqref{btp} and \eqref{al}. Specifically, with the initial values given in the above Theorem, we iterate \eqref{diffpn} for $n=1,\ldots,8$ to derive the expressions for $\{\textbf{p}(j)\}_{j=3}^{10}$ in sequence. Based on this, we apply \eqref{al} for $n=2,\ldots,9$ and \eqref{btp} for $n=1,\ldots,9$ to deduce the expressions for $\{\alpha_j\}_{j=2}^9$ and $\{\beta_j\}_{j=1}^9$. Note that $\alpha_0$ and $\alpha_1$ are given in \eqref{al0bt1}.
\end{remark}

According to Table \ref{Table-1}, we propose the following conjecture.
\begin{conjecture}\label{Cw1-albtp}
$\{\alpha_n\}_{n\geq4}$ and $\{\beta_n\}_{n\geq 2}$ have the following symbolic forms:
\begin{subequations}
\begin{align}
\alpha_{2k}=&\frac{2(\lambda+2k)}{\zeta}+k\zeta\left((2\lambda+2k-1)\frac{X_{k^2-1}(\zeta^2)}{X_{k^2}(\zeta^2)}-(2\lambda+2k+1)\frac{X_{k(k+1)-1}(\zeta^2)}{2X_{k(k+1)}(\zeta^2)}\right),\label{alsf-1}\\ \alpha_{2k+1}=&\frac{2(\lambda+2k+1)}{\zeta}+(2\lambda+2k+1)\zeta\left(k\frac{X_{k(k+1)-1}(\zeta^2)}{2X_{k(k+1)}(\zeta^2)}-(k+1)\frac{X_{(k+1)^2-1}(\zeta^2)}{X_{(k+1)^2}(\zeta^2)}\right),\label{alsf-2}
\end{align}
\end{subequations}
for $k\geq2$, and
\begin{align*}
\beta_{2k} = -2k(2\lambda+2k-1)\frac{X_{k(k-1)}(\zeta^2)\cdot X_{k(k+1)}(\zeta^2)}{\zeta^2 \left(X_{k^2}(\zeta^2)\right)^2}, \qquad\beta_{2k+1} = \frac{X_{k^2}(\zeta^2)\cdot X_{(k+1)^2}(\zeta^2)}{\zeta^2 \left(X_{k(k+1)}(\zeta^2)\right)^2},
\end{align*}
for $k\geq1$. Here $X_j(\cdot)$ is a monic polynomial of degree $j$, with $X_0(\zeta^2):=1$ and $X_1(\zeta^2):=\zeta^2-2\lambda$. Moreover, for $\{\textbf{p}(n)\}_{n\geq4}$, we conjecture that
\begin{align*}
\textbf{p}(2k)=&-ik(2\lambda+2k-1)\left(\frac{2}{\zeta}-\zeta\frac{X_{k^2-1}(\zeta^2)}{X_{k^2}(\zeta^2)}\right),\\
\textbf{p}(2k+1)=&-i\left(\frac{2}{\zeta}(2k+1)(\lambda+k)-k(2\lambda+2k+1)\zeta\frac{X_{k(k+1)-1}(\zeta^2)}{2X_{k(k+1)}(\zeta^2)}\right), \end{align*}
with $k\geq2$.
\end{conjecture}
\begin{remark}
According to the conjectured properties of a certain Hankel determinant, the symbolic forms of $\alpha_n$ and $\beta_n$ were also provided in \cite{MCM} (below Table 1 therein). They are given below:
\begin{align}
\alpha_{2k} =& \frac{V_{k(2k+1)}(\zeta^2)}{\zeta S_{k^2}(\zeta^2)\cdot S_{k(k+1)}(\zeta^2)}, & \alpha_{2k+1} =& \frac{V_{(k+1)(2k+1)}(\zeta^2)}{\zeta S_{k(k+1)}(\zeta^2)\cdot S_{(k+1)^2}(\zeta^2)},\label{alsf-3}\\
\beta_{2k} =& -\frac{S_{k(k-1)}(\zeta^2)\cdot S_{k(k+1)}(\zeta^2)}{\zeta^2 \left(S_{k^2}(\zeta^2)\right)^2},& \beta_{2k+1} =& -\frac{S_{k^2}(\zeta^2)\cdot S_{(k+1)^2}(\zeta^2)}{\zeta^2 \left(S_{k(k+1)}(\zeta^2)\right)^2},\nonumber
\end{align}
which hold for $k\geq0$, except for $k=0$ in $\beta_{2k}$. Here $V_j(\cdot)$ and $S_j(\cdot)$ are polynomials of degree $j$, but their leading coefficients were not specified. We find from \cite{MCM} that $S_0=1$.
Our expressions of $\alpha_n$ and $\beta_n$ presented in Table \ref{Table-1} are consistent with the above conjecture where
\begin{align*}
S_1(\zeta^2)=&-(\zeta^2-2\lambda),&S_2(\zeta^2)=&2\kappa_1X_2,\\
S_4(\zeta^2)=&4\kappa_1^2X_4,&S_6(\zeta^2)=&32 \kappa_1^3\kappa_3X_6,\\ S_9(\zeta^2)=&-256\kappa_1^4\kappa_3^2X_9,&S_{12}(\zeta^2)=&12288\kappa_1^5\kappa_3^3\kappa_5X_{12},\\
S_{16}(\zeta^2)=&589824\kappa_1^6\kappa_3^4\kappa_5^2X_{16},&S_{20}(\zeta^2)=&226492416\kappa_1^7\kappa_3^5\kappa_5^3\kappa_7X_{20},\\
S_{25}(\zeta^2)=&-86973087744\kappa_1^8\kappa_3^6\kappa_5^4\kappa_7^2X_{25},&&
\end{align*}
with $\kappa_j:=2\lambda+j$, and
\begin{gather*}
\quad V_0(\zeta^2)=2\lambda, \qquad\qquad\qquad\qquad\qquad\qquad\qquad V_1(\zeta^2)=-(\zeta^2-4\lambda^2-4\lambda),\qquad\qquad\\
V_3(\zeta^2)=-(2\lambda+1)\left( (6 \lambda +7)\zeta ^6-4 \lambda (\lambda+4)(2 \lambda+3)\zeta^4+32 \lambda ^2
   (2\lambda +3)\zeta ^2-32 \lambda ^3 (\lambda +2)\right).
\end{gather*}
The explicit expressions of $V_j$ with $j=6,10,15,21,28,36$ are omitted here for brevity.
From the above formulas for $S_j$, we conjecture that
\begin{align*}
S_{k^2}(\zeta^2)=&X_{k^2}(\zeta^2)\cdot(-1)^k 2^{k(k-1)}\prod_{\ell=1}^{k-1}\left(\ell(2\lambda+2\ell-1)\right)^{2(k-\ell)},\qquad k\geq2,\\
S_{k(k+1)}(\zeta^2)=&X_{k(k+1)}(\zeta^2)\cdot 2^{k^2}\prod_{\ell=1}^{k}\left(\ell(2\lambda+2\ell-1)\right)^{2(k-\ell)+1},\qquad k\geq1.
\end{align*}
Bringing the right-hand sides of both \eqref{alsf-1} and \eqref{alsf-2} to a common denominator, and comparing them with \eqref{alsf-3}, in view of the above expressions for $S_{k^2}(\zeta^2)$ and $S_{k(k+1)}(\zeta^2)$, we obtain the forms of $V_{k(2k+1)}(\zeta^2)$ and $V_{(k+1)(2k+1)}(\zeta^2)$, which are not displayed for brevity.
\end{remark}

\begin{remark}
When $\lambda=1$ and $\zeta=8.653727912911012$ {\rm(}a zero of $J_0${\rm)}, the numerical results for
$\alpha_n$ and $\beta_n$ obtained from \eqref{diff2}-\eqref{diff1} are
shown in Tables \ref{tab:comparison-al} and \ref{tab:comparison-bt} (see Appendix \ref{table-value} for the Mathematica code). The values of $\alpha_n$ for
$n = 1,\dots,8$ and of $\beta_n$ for $n = 1,\dots,11$ agree with those in \cite{MCS} to within a relative error of order
$10^{-9}$ or less, which confirms the correctness of the derivation of \eqref{diff1}-\eqref{diff2}. For larger $n$, the forward iteration becomes numerically
unstable: the recurrence tends to amplify rounding errors, which eventually
dominate the result. This instability is inherent in the forward iteration
scheme; extending the computation to larger $n$ by numerically stable methods is left for future work.
\begin{table}[htbp]
\centering
\caption{Comparison of $\alpha_n$ with those in \cite{MCS}}
		\label{tab:comparison-al}
		\begin{tabular}{|c|c|c|c|}
			\hline
			$n$ & $\alpha_n (\text{present})$& $\alpha_n (\text{Ref.~\cite{MCS}})$ & \text{Relative error} \\\hline
			0 & 2.311142689171080(-1) &  2.311142689171081(-1) & 3.60(-16) \\\hline
			1 & 1.060445637373780(-1) &  1.060445637373780(-1) & 0
			\\\hline
			2 & 7.651132178528472(-1) & 7.651132178528472(-1) & 0
			\\\hline
			3 & -6.295009390649609(-2) & -6.295009390649605(-2) & 4.41(-16)
			\\\hline
			4 & 1.932207097354235 & 1.932207097354234 & 4.60(-16)
			\\\hline
			5 & -8.605888996725330(-1) &  -8.605888996725327(-1) & 3.87(-16)
			\\\hline
			6 & 5.664110881321793(-2) & 5.664110881321379(-2) & 7.30(-14)
			\\\hline
			7 & -4.402172898241949(-3) & -4.402172898193667(-3) & 1.10(-11)
			\\\hline
			8 & 2.654519071151190(-4) & 2.654519062269625(-4) & 3.35(-9)
			\\\hline
			9 &-1.308068140395573(-5) & -1.308066174957451(-5) & 1.50(-6)
			\\\hline
			10 & 5.351786821724559(-7) & 5.346570353604267(-7) & 9.76(-4)
			\\\hline
			11 & -3.475719208785617(-8) & -1.839775190533208(-8) & -
			\\\hline
			12 & 5.994112736617820(-7) & 5.397928683729530(-10) & -
			\\\hline
			13 & -2.531964264266585(-5) & -1.365470127894177(-11) & -
			\\\hline
			14 & 1.224455489487730(-3) & 3.007166162157992(-13) & -
			\\\hline
		\end{tabular}
\parbox{0.78\textwidth}{\footnotesize \textbf{Note.} Relative errors are computed as $|\alpha_n(\text{present}) - \alpha_n (\text{Ref.~\cite{MCS}})|/|\alpha_n (\text{Ref.~\cite{MCS}})|$. \\Numbers in parentheses denote powers of ten, e.g., $3.60(-16)$ stands for $3.60 \times 10^{-16}$.}
	\end{table}

\begin{table}[htbp]
\centering
\caption{Comparison of $\beta_n$ with those in \cite{MCS}}
\label{tab:comparison-bt}
\begin{tabular}{|c|c|c|c|}
\hline
$n$  & $\beta_n (\text{present})$ & $\beta_n (\text{Ref.~\cite{MCS}~})$ & \text{Relative error} \\\hline
	0  & $\beta_0$ & - & - \\\hline
		1  & 9.732930973514550(-1) &  9.732930973514553(-1) & 2.28(-16)  \\\hline
		2  & -7.673248270228838(-2) & -7.673248270228832(-2) & 5.43(-16)\\\hline
		3  & 9.653084647774710(-1) &  9.653084647774708(-1) & 1.15(-16)\\\hline
		4 & -1.506275573638138(-1) & -1.506275573638140(-1) & 1.29(-15)\\\hline
		5  & 2.187766221402498 &  2.187766221402496 & 8.12(-16)\\\hline
		6 & 1.588977014752850(-1) & 1.588977014752849(-1) & 5.24(-16)\\\hline
		7  & 2.582628598423706(-1) &  2.582628598423771(-1) & 2.54(-14)\\\hline
		8  & 2.494455025987613(-1) & 2.494455025986602(-1) & 4.05(-13)\\\hline
		9  & 2.500301915254477(-1) &  2.500301915274877(-1) & 8.16(-12)\\\hline
		10 & 2.499986478416937(-1) & 2.499986477921526(-1) & 1.98(-10)\\\hline
		11  & 2.500000491940501(-1) & 2.500000506256831(-1) & 5.73(-9)\\\hline
		12  & 2.500000469704600(-1) & 2.499999983937877(-1) & 1.94(-7)\\\hline
		13 & 2.499980866275262(-1) & 2.500000000437017(-1) & 7.65(-6)\\\hline
		14 & 2.500866293625434(-1) &  2.499999999989697(-1) & 3.47(-4)\\\hline
	\end{tabular}
		\end{table}
\end{remark}

When $\lambda=0$ in \eqref{w1}, the weight function becomes
\[w_1(x)=(1-x^2)^{-1/2}\exp(i\zeta x), \qquad x\in[-1,1],\]
where $\zeta\in\mathbb{R}\setminus\{0\}$ is a zero of $J_{-1}$. Using the initial values \eqref{al0bt1} and iterating \eqref{diff1}-\eqref{diff2}, we obtain the expressions for $\alpha_n,\beta_n,$ and then those for  $\textbf{p}(n)$. Their forms, which are quite different from those in the general case, are presented in Table \ref{Table-1-0} for $n=0,1,\ldots,13$.
\begin{table}[htbp]
	\centering
	 \renewcommand{\arraystretch}{1.4}
	\setlength{\tabcolsep}{3pt}
\caption{$\{\alpha_n,\beta_n, \textbf{p}(n)\}_{n=0}^{13}$ for $\lambda=0$ with $\beta_0$ arbitrary}\label{Table-1-0}
\begin{tabular}{|c|c|c|c|}
\hline
$n$ & $\alpha_n$ & $\beta_n$ & $p(n)$  \\
\hline
0 & 0 & $\beta_{0}$ & $0$ \\
\hline
1 & $\frac{1}{\zeta}$ & $1$ & 0 \\
\hline
2 & $\frac{7}{2\zeta}$ & $-\frac{2}{\zeta^2}$ & $-\frac{\mathit{i}}{\zeta}$ \\
\hline
3 & $\frac{3^3\cdot3}{14\zeta}-\frac{30\zeta}{7(\zeta^2-21/4)}$ & $\frac{\zeta^2-21/4}{\zeta^2}$ & $-\frac{9\mathit{i}}{2\zeta}$ \\
\hline
4 &
$\frac{2^5\cdot19}{7\cdot11\zeta}+5\zeta\left(\frac{6}{7(\zeta^2-21/4)}-\frac{7(\zeta^2-177/8)}{11Y_{2}}\right)$ & $-\frac{12Y_{2}}{\zeta^2(\zeta^2-21/4)^2}$ & $-\frac{6\mathit{i}}{7}\left(\frac{12}{\zeta}-\frac{5\zeta}{\zeta^2-21/4}\right)$ \\
\hline
5 & $\frac{5^3\cdot7}{8\cdot11\zeta}+35\zeta\left(\frac{\zeta^2-177/8}{11Y_{2}}-\frac{3Y_{3}}{8Y_{4}}\right)$ & $\frac{(\zeta^2-21/4)Y_{4}}{\zeta^2Y_{2}^2}$ & $-\frac{5\mathit{i}}{11}\left(\frac{40}{\zeta}-\frac{7\zeta(\zeta^2-177/8)}{Y_{2}}\right)$ \\
\hline
6 & $\frac{3^4\cdot13}{8\cdot11\zeta}+21\zeta\left(\frac{5Y_{3}}{8Y_{4}}-\frac{9Y_{5}}{22Y_{6}}\right)$ & $-\frac{30Y_{2}Y_{6}}{\zeta^2Y_{4}^2}$ & $-\frac{15\mathit{i}}{8}\left(\frac{15}{\zeta}-\frac{7\zeta Y_{3}}{Y_{4}}\right)$ \\
\hline
7 & $\frac{7^3\cdot13}{11\cdot29\zeta}+63\zeta\left(\frac{3Y_{5}}{22Y_{6}}-\frac{12Y_{8}}{29Y_{9}}\right)$ & $\frac{Y_{4}Y_{9}}{\zeta^2Y_{6}^2}$ & $-\frac{63\mathit{i}}{11}\left(\frac{7}{\zeta}-\frac{3\zeta Y_{5}}{2Y_{6}}\right)$ \\
\hline
8 &$\frac{4^4\cdot67}{29\cdot37\zeta}+27\zeta\left(\frac{28Y_{8}}{29Y_{9}}-\frac{22Y_{11}}{37Y_{12}}\right)$ & $-\frac{56Y_{6}Y_{12}}{\zeta^2Y_{9}^2}$ & $-\frac{28\mathit{i}}{29}\left(\frac{56}{\zeta}-\frac{27\zeta Y_{8}}{Y_{9}}\right)$ \\
\hline
9 & $\frac{9^3\cdot21}{23\cdot37\zeta}+198\zeta\left(\frac{3Y_{11}}{37Y_{12}}-\frac{5Y_{15}}{23Y_{16}}\right)$ & $\frac{Y_{9}Y_{16}}{\zeta^2Y_{12}^2}$ & $-\frac{54\mathit{i}}{37}\left(\frac{48}{\zeta}-\frac{11\zeta Y_{11}}{Y_{12}}\right)$ \\
\hline
10 & $\frac{5^3\cdot103}{23\cdot28\zeta}+5\zeta\left(\frac{198Y_{15}}{23Y_{16}}-\frac{143Y_{19}}{28Y_{20}}\right)$ & $-\frac{90Y_{12}Y_{20}}{\zeta^2Y_{16}^2}$ & $-\frac{45\mathit{i}}{23}\left(\frac{45}{\zeta}-\frac{22\zeta Y_{15}}{Y_{16}}\right)$ \\
\hline
11 & $\frac{11^3\cdot31}{28\cdot67\zeta}+715\zeta\left(\frac{Y_{19}}{28Y_{20}}-\frac{6Y_{24}}{67Y_{25}}\right)$ & $\frac{Y_{16}Y_{25}}{\zeta^2Y_{20}^2}$ & $-\frac{55\mathit{i}}{28}\left(\frac{55}{\zeta}-\frac{13\zeta Y_{19}}{Y_{20}}\right)$ \\
\hline
12 & $\frac{6^3\cdot588}{67\cdot79\zeta}+195\zeta\left(\frac{22Y_{24}}{67Y_{25}}-\frac{15Y_{29}}{79Y_{30}}\right)$ & $-\frac{132Y_{20}Y_{30}}{\zeta^2Y_{25}^2}$ & $-\frac{66\mathit{i}}{67}\left(\frac{132}{\zeta}-\frac{65\zeta Y_{24}}{Y_{25}}\right)$ \\
\hline
13 & $\frac{13^3\cdot43}{46\cdot79\zeta}+585\zeta\left(\frac{5Y_{29}}{79Y_{30}}-\frac{7Y_{35}}{46Y_{36}}\right)$ & $\frac{Y_{25}Y_{36}}{\zeta^2Y_{30}^2}$ & $\frac{-117\mathit{i}}{79}\left(\frac{104}{\zeta}-\frac{25\zeta Y_{29}}{Y_{30}}\right)$ \\
\hline
\end{tabular}

\vspace{1mm}
\parbox{0.85\textwidth}{\footnotesize \textbf{Note.} $Y_j$ stands for $Y_j(\zeta^2)$, with $Y_j(\cdot)$ denoting a monic polynomial of degree $j$. Hence, $Y_j(\zeta^2)$ is a monic polynomial of degree $2j$ in $\zeta$, containing only even powers of $\zeta$. The expressions of $Y_j(\zeta^2)$ for $j=2,3,4,5,6,8,9,11,12,15,16$ are given in Appendix \ref{symbols},
with the others omitted due to their length.}
\end{table}

According to Table \ref{Table-1-0}, we make the following conjecture for the symbolic forms of $\alpha_n,\beta_n,$ and $\textbf{p}(n)$.

\begin{conjecture}
For $\{\alpha_n\}_{n\geq6}$ and $\{\beta_n\}_{n\geq4}$, we conjecture that
\begin{align*}
\alpha_{2k}=&\frac{4k^3(4k^2+3)}{(k(2k-1)+1)(k(2k+1)+1)\zeta}\\
&+k(k-1)(2k+1)\zeta\left(\frac{2k-1}{k(2k-1)+1}\cdot\frac{Y_{(k-1)^2-1}(\zeta^2)}{Y_{(k-1)^2}(\zeta^2)}-\frac{2k+3}{2(k(2k+1)+1)}\cdot\frac{Y_{k(k-1)-1}(\zeta^2)}{Y_{k(k-1)}(\zeta^2)}\right),\\
\alpha_{2k+1}=&\frac{2(2k+1)^3(k(k+1)+1)}{(2k^2+3k+2)(k(2k+1)+1)\zeta}\\
&+k(2k+1)(2k+3)\zeta \left(\frac{k-1 }{k(2k+1)+1}\cdot\frac{Y_{k(k-1)-1}(\zeta^2)}{2Y_{k(k-1)}(\zeta^2)}-\frac{k+1 }{(k+1)(2k+1)+1}\cdot\frac{Y_{k^2-1}(\zeta^2)}{Y_{k^2}(\zeta^2)}\right),
\end{align*}
with $k\geq3$, and
\begin{align*}
\beta_{2k} =& -2k(2k-1)\frac{Y_{(k-1)(k-2)}(\zeta^2)\cdot Y_{k(k-1)}(\zeta^2)}{\zeta^2 \left(Y_{(k-1)^2}(\zeta^2)\right)^2},& \beta_{2k+1} =& \frac{Y_{(k-1)^2}(\zeta^2)\cdot Y_{k^2}(\zeta^2)}{\zeta^2 \left(Y_{k(k-1)}(\zeta^2)\right)^2},
\end{align*}
with $k\geq2$. Here $Y_j(\cdot)$ is a monic polynomial of degree $j$, with $Y_0(\zeta^2):=1$ and $Y_1(\zeta^2):=\zeta^2-21/4$. Regarding the sequence $\{\textbf{p}(n)\}_{n\geq6}$, we conjecture the following:
\begin{align*}
\textbf{p}(2k)=&-\frac{ik(2k-1)}{k(2k-1)+1}\left(\frac{2k(2k-1)}{\zeta}- (k-1)(2k+1)\zeta \frac{Y_{(k-1)^2-1}(\zeta^2)}{Y_{(k-1)^2}(\zeta^2)}\right),\\
\textbf{p}(2k+1)=&-\frac{ik(2k+1)}{k(2k+1)+1}\left(\frac{2k(2k+1)}{\zeta}-(k-1)(2k+3)\zeta\frac{Y_{k(k-1)-1}(\zeta^2)}{2Y_{k(k-1)}(\zeta^2)}\right), \end{align*}
for $k\geq3$.

In addition, the quantities $\{Y_j\}$ appearing above are conjectured to be related to the $\{X_j\}$ of Conjecture \ref{Cw1-albtp} {\rm(}with $\lambda=0${\rm)} by
\begin{align*}
Y_{(k-1)^2-1}=&\frac{(k(2k-1)+1)\zeta X_{k^2-1}-2 X_{k^2}/\zeta}{(k-1)(2k+1)\zeta^{4k-1}},&Y_{(k-1)^2}=&\frac{X_{k^2}}{\zeta^{4k-2}},\\
Y_{k(k-1)-1}=&\frac{(k(2k+1)+1)\zeta X_{k(k+1)-1}-4 X_{k(k+1)}/\zeta}{(k-1)(2k+3)\zeta^{4k+1}},&Y_{k(k-1)}=&\frac{X_{k(k+1)}}{\zeta^{4k}},
\end{align*}
for $k\geq3$. Note that $Y_2(\zeta^2)=X_6(\zeta^2)/\zeta^8$.

\end{conjecture}

\section{$x-$Modified Oscillatory Gegenbauer Weight and Difference Equations}
For the weight function given by \eqref{w2}, we have
\begin{align}
v(x)=-\ln w_2(x)=-\ln x-\left(\gamma-1/2\right)\ln (1-x^{2})-\mathit{i}\eta x,
\end{align}
so that
\begin{align}
v'(x)=-\frac{1}{x}+\frac{(2\gamma-1)x}{1-x^{2}}-\mathit{i}\eta.
\end{align}
Consequently,
\begin{align*}
\frac{(1-z^{2})v'(z)-(1-x^{2})v'(x)}{z-x}=\frac{1}{zx}+2\gamma+ \mathit{i}\eta x+\mathit{i}\eta z.
\end{align*}
Substituting it into \eqref{An} and \eqref{Bn} , we get the expressions for $ A_{n}(z) $ and $ B_{n}(z) $ .
\begin{lemma}
We have
	\begin{align}
A_{n}(z)=&\frac{\mathit{i}\eta z+2n+2\gamma+1-\eta\alpha_n}{1-z^{2}} +\frac{R_{n}(\zeta)}{z(1-z^{2})},\qquad n\geq0,\label{(2.8)}\\
B_{n}(z)=&\frac{nz+\mathit{i}\eta\beta_n-\textbf{p}(n)}{1-z^{2}} +\frac{r_{n}(\zeta)}{z(1-z^{2})},\qquad n\geq1, \label{(2.9)}
\end{align}
where
\begin{align}
R_{n}(\zeta):=&\frac{1}{h_{n}}\int_{-1}^{1}P_{n}^2(x)\left(1-x^{2}\right)^{\gamma-1/2}\exp(\mathit{i}\eta x)\,dx ,\qquad n\geq0,\\
r_{n}(\zeta):=&\frac{1}{h_{n-1}}\int_{-1}^{1}P_{n}(x)P_{n-1}(x)\left(1-x^{2}\right)^{\gamma-1/2}\exp(\mathit{i}\eta x)\,dx,\qquad n\geq1.
\end{align}
\end{lemma}

Inserting \eqref{(2.8)}-\eqref{(2.9)} into {$(S_1)$}, and equating the coefficients of
$z^{-1}$, $(z-1)^{-1}$ and $(z+1)^{-1}$ on both sides, we obtain
\begin{gather}
r_{n+1}+r_{n}=1-\mathit{i}\alpha_{n}R_{n},\label{(2.12)}
\end{gather}
\begin{gather}
	\begin{split}
&-\mathit{i}\eta(\beta_{n+1}+\beta_{n})+\textbf{p}(n+1)+\textbf{p}(n)-2n-1-r_{n+1}-r_{n}\\
&\qquad\qquad=2\gamma-1-(1-\mathit{i}\alpha_{n})(2n+2\gamma+1-\eta\alpha_{n}+\mathit{i}\eta+R_{n}),\label{(2.13)}
	\end{split}
\end{gather}
\begin{gather}
	\begin{split}
&\mathit{i}\eta(\beta_{n+1}+\beta_{n})-\textbf{p}(n+1)-\textbf{p}(n)-2n-1-r_{n+1}-r_{n}\\
&\qquad\qquad=2\gamma-1-(1+\mathit{i}\alpha_{n})(2n+2\gamma+1-\eta\alpha_{n}-\mathit{i}\eta-R_{n}),\label{(2.14)}
	\end{split}
\end{gather}
for $n\geq1$, in accordance with the validity range of \eqref{(2.9)}. From $(S_2)$, we find
\begin{gather}
-\mathit{i}\alpha_{n}(r_{n+1}-r_{n})=\beta_{n+1}R_{n+1}-\beta_{n}R_{n-1},\label{(2.16)}
\end{gather}
\begin{gather}
\begin{split}
&\left(1-\mathit{i}\alpha_{n}\right)\big(\mathit{i}\eta(\beta_{n+1}-\beta_{n})-\textbf{p}(n+1)+\textbf{p}(n)+r_{n+1}-r_{n}+1\big)\\
&\;=\beta_{n+1}\left(2n+2\gamma+3-\eta\alpha_{n+1}+\mathit{i}\eta+R_{n+1}\right)-\beta_{n}\left(2n+2\gamma-1-\eta\alpha_{n-1}+\mathit{i}\eta+R_{n-1}\right),\label{(2.17)}
\end{split}
\end{gather}
\begin{gather}
	\begin{split}
&-\left(1+\mathit{i}\alpha_{n}\right)\big(\mathit{i}\eta(\beta_{n+1}-\beta_{n})-\textbf{p}(n+1)+\textbf{p}(n)-r_{n+1}+r_{n}-1\big)\\
&\;=\beta_{n+1}\left(2n+2\gamma+3-\eta\alpha_{n+1}-\mathit{i}\eta-R_{n+1}\right)-\beta_{n}\left(2n+2\gamma-1-\eta\alpha_{n-1}-\mathit{i}\eta-R_{n-1}\right),\label{(2.18)}
\end{split}
\end{gather}
for $n\geq1$.

Using \eqref{(2.12)}-\eqref{(2.14)} and \eqref{(2.17)}-\eqref{(2.18)}, we express $R_n$ and $r_n$ in terms of $\alpha_n,\beta_n,$ and $\textbf{p}(n)$.

\begin{lemma}
	The auxiliary quantities $R_{n}$ and $r_{n}$ admit the following expressions:
	\begin{align}
R_{n}=&\mathit{i}\alpha_{n}(2n+2\gamma+2-\eta\alpha_{n})+\mathit{i}\eta(\beta_{n+1}+\beta_{n}-1)-2\textbf{p}(n), \label{(3.1)}
\end{align}
\begin{align}
		\begin{split}
2r_{n}=&2+2\mathit{i}\alpha_{n}\textbf{p}(n)+\alpha_{n}^{2}(2n+2\gamma+3-\eta\alpha_{n})-\eta\alpha_{n}\\
&-\beta_{n+1}(2n+2\gamma+3-\eta\alpha_{n+1}-2\eta\alpha_{n})+\beta_{n}(2n+2\gamma-1-\eta\alpha_{n-1}),\label{(3.2)}
		\end{split}
	\end{align}
for $n\geq1$.
\end{lemma}

\begin{proof}
Subtracting \eqref{(2.14)} from \eqref{(2.13)}, and replacing $\textbf{p}(n+1)$ in the resulting identity using \eqref{al}, we come to \eqref{(3.1)}.

Adding \eqref{(2.17)} and \eqref{(2.18)}, and using \eqref{al} to eliminate $\textbf{p}(n+1)$ in the obtained equation, we get
\begin{align}\label{s2-4}
r_{n+1}-r_{n}=\beta_{n+1}\left(2n+2\gamma+3-\eta\alpha_{n+1}-\eta\alpha_{n}\right)-\beta_{n}\left(2n+2\gamma-1-\eta\alpha_{n}-\eta\alpha_{n-1}\right)-\alpha_{n}^{2}-1.
\end{align}	
Subtracting it from \eqref{(2.12)} and replacing $R_n$ via \eqref{(3.1)}, we deduce \eqref{(3.2)}.
\end{proof}

Using \eqref{(3.1)}-\eqref{(3.2)} to eliminate the auxiliary quantities from \eqref{(2.12)} and from the equation obtained by substituting \eqref{s2-4} into \eqref{(2.16)}, we arrive at the following result.
	\begin{proposition}
		$\textbf{p}(n)$ and the recurrence coefficients satisfy the following coupled difference equations:
	 \begin{align}\label{diffpn-1}
			 \textbf{p}(n)=\textbf{p}(n-1)-\mathit{i}\alpha_{n-1},
	\end{align}
\begin{equation}\label{diffal-1}
\begin{aligned} \alpha_{n+1}=&\frac{1}{\eta\beta_{n+1}}\left[\left(\alpha_{n-1}-\alpha_{n}\right)\left(2\mathit{i}\textbf{p}(n)-\eta\right)+\left(\eta\alpha_{n}-2\gamma_{n}-3\right)\alpha_{n}^2-\left(\eta\alpha_{n-1}-2\gamma_{n}+3\right)\alpha_{n-1}^2\right.\\
&\left.\qquad\quad			 +\beta_{n}\big(\eta(\alpha_{n-1}-\alpha_{n})+2\big)+\beta_{n-1}\big(\eta (2\alpha_{n-1}+\alpha_{n-2})-2\gamma_{n}+3\big)-2\right]\\
&-2\alpha_{n}+\frac{2\gamma_{n}+3}{\eta},
\end{aligned}
\end{equation}
\begin{equation}\label{diffbt-1}
\begin{aligned}
\beta_{n+2}=& \frac{1}{\eta\beta_{n+1}}\left\{\beta_{n}\left[(2\gamma_{n}-1)\alpha_{n}+2(\gamma_{n}-1)\alpha_{n-1}+\eta\left(\beta_{n}+\beta_{n-1}-\alpha_{n}^2-\alpha_{n-1}^2-\alpha_{n-1}\alpha_{n}-1\right)\right]\right.\\
&\left.\qquad\quad+\alpha_{n}(\alpha_{n}^2+1)+2i\textbf{p}(n)\left(\beta_{n}-\beta_{n+1}\right)\right\}+\alpha_{n+1}^2+\alpha_{n}^2+\alpha_{n+1}\alpha_{n}\\
&-\frac{1}{\eta}\left[2(\gamma_{n}+2)\alpha_{n+1}+(2\gamma_{n}+5)\alpha_{n}\right]-\beta_{n+1}+1,
\end{aligned}
\end{equation}
for $n\geq2$ with $\gamma_{n}:=n+\gamma$, which can be iterated in $n$ using the initial values $\textbf{p}(1),$ $\alpha_0,\alpha_1,\alpha_2,$ and $\beta_1,\beta_2,\beta_3$.
	\end{proposition}
	\begin{proof}
Equation \eqref{diffpn-1} is actually  \eqref{al} with $n$ replaced by $n-1$, which holds for $n\geq1$.

Replacing $n$ with $n-1$ in \eqref{(2.12)} gives
		\begin{align*}
			 r_{n}+r_{n-1}=1-\mathit{i}\alpha_{n-1}R_{n-1},\qquad n\geq2.
		\end{align*}	
Using \eqref{(3.2)} and \eqref{(3.1)} to eliminate $\{r_n,r_{n-1}\}$ and $R_{n-1}$ from this identity respectively, we get an equation involving the recurrence coefficients, $\textbf{p}(n),$ and $\textbf{p}(n-1)$. With $\textbf{p}(n-1)$ replaced by $\mathit{i}\alpha_{n-1}+\textbf{p}(n)$ in this equation, we come to \eqref{diffal-1}.

Plugging \eqref{s2-4} into \eqref{(2.16)}, and using \eqref{(3.1)} to eliminate $R_{n+1}$ (for $n\geq0$) and $R_{n-1}$ (for $n\geq2$), we obtain an equation involving the recurrence coefficients and $ \textbf{p}(n\pm1)$, holding for $n\geq2$. By replacing $\textbf{p}(n+1)$ with  $\textbf{p}(n)-\mathit{i}\alpha_{n}$ and $\textbf{p}(n-1)$ with $\textbf{p}(n)+\mathit{i}\alpha_{n-1}$ in this equation, we derive \eqref{diffbt-1} which is valid for $n\geq2$.
	\end{proof}

Eliminating $\textbf{p}(n)$ from \eqref{diffpn-1}-\eqref{diffbt-1} leads us to the theorem below.
\begin{theorem}
The recurrence coefficients satisfy the following coupled difference equations:
\begin{equation}\label{diff-bt}
\begin{aligned}		 \beta_{n+1}=&\alpha_{n}^2+\alpha_{n}\alpha_{n-1}-\alpha_{n-1}\alpha_{n-2}-\alpha_{n-2}^2-\frac{1}{\eta}\left(2(\gamma_n+1)\alpha_{n}+5\alpha_{n-1}-2(\gamma_n-1)\alpha_{n-2}\right)\\
&-\beta_{n}+\beta_{n-1}+\frac{1}{\eta\beta_{n}}\left[\beta_{n-1}\left(\eta\beta_{n-2}-2\alpha_{n-2}-\alpha_{n-1}\right)+\alpha_{n-1}\left(\alpha_{n-1}^2+1\right)\right]\\
&+\frac{\beta_{n}-\beta_{n-1}}{\eta\beta_{n}(\alpha_{n-1}-\alpha_{n-2})}\left[3\left(\alpha_{n-1}^2+\alpha_{n-2}^2\right)-\beta_{n}\left(2\gamma_n+1-\eta(\alpha_{n}+2\alpha_{n-1})\right)-2\beta_{n-1}\right.\\
&\left.\qquad\qquad\qquad\qquad\qquad+\beta_{n-2}\left(2\gamma_n-5-\eta( 2\alpha_{n-2}+\alpha_{n-3})\right)+2\right],
	\end{aligned}
\end{equation}
\begin{equation}\label{diff-al}
\begin{aligned}
\alpha_{n+1}=&\frac{1}{\eta\beta_{n+1}}\left\{-3\alpha_{n}^2-2\alpha_{n}\alpha_{n-1}-\alpha_{n-1}^2+2(\beta_{n}-1)+\beta_{n-1}(\eta(2\alpha_{n-1}+\alpha_{n-2})-2\gamma_{n}+3)\right.\\
&\left.\qquad\quad\,+(\alpha_{n-1}-\alpha_{n})\big[(\alpha_{n-2}-\alpha_{n})\left(\eta(\alpha_{n}+\alpha_{n-1}+\alpha_{n-2})-2\gamma_n\right)+\eta(\beta_n-\beta_{n-1})\big]\right.\\
&\left.\qquad\quad\,+\frac{\alpha_{n-1}-\alpha_{n}}{\alpha_{n-2}-\alpha_{n-1}}\left[\alpha_{n-1}^2+5\alpha_{n-2}^2+2+\beta_{n}\left(\eta(\alpha_{n}+2\alpha_{n-1})-2\gamma_{n}-1\right)\right.\right.\\
&\left.\left.\quad\qquad\qquad\qquad\qquad\quad\,-2\beta_{n-1}+\beta_{n-2}\left(2\gamma_{n}-5-\eta(2\alpha_{n-2}+\alpha_{n-3})\right)\right]\right\}\\
&-2\alpha_{n}+\frac{2\gamma_{n}+3}{\eta},
\end{aligned}
\end{equation}
for $n\geq3$ with $\gamma_{n}=n+\gamma$, which can be iterated in $n$ given the initial values $\{\alpha_0,\alpha_1,\alpha_2,\alpha_3\}$ and $\{ \beta_1,\beta_2,\beta_3\}$.
	\end{theorem}
	\begin{proof}
Solving for $\textbf{p}(n)$ from \eqref{diffal-1} yields
\begin{equation}\label{pnalbt-1}
\begin{aligned}
2\mathit{i}\textbf{p}(n)=&-2\gamma_n(\alpha_n+\alpha_{n-1})+\eta(\alpha_n^2+\alpha_n\alpha_{n-1}+\alpha_{n-1}^2-\beta_n+1)\\
&+\frac{1}{\alpha_{n-1}-\alpha_{n}}\left[3\left(\alpha_{n}^2+\alpha_{n-1}^2\right)-\beta_{n+1}\left(2\gamma_n+3-\eta(\alpha_{n+1}+2\alpha_{n})\right)\right.\\
&\left.\qquad\qquad\qquad\;-2\beta_n+\beta_{n-1}\left(2\gamma_n-3-\eta( 2\alpha_{n-1}+\alpha_{n-2})\right)+2\right],\qquad n\geq2.
\end{aligned}
\end{equation}
Replacing $n$ with $n-1$ in this expression, and subtracting the resulting identity from the above one, in view of $\mathit{i}(\textbf{p}(n)-\textbf{p}(n-1))=\alpha_{n-1}$, we come to \eqref{diff-al}, valid for $n\geq3$.

Solving \eqref{diffbt-1} for $\textbf{p}(n)$  produces
\begin{equation*}
\begin{aligned}		 2\mathit{i}\textbf{p}(n)=&\frac{1}{\beta_{n+1}-\beta_{n}}\left\{\beta_{n+1}\left[\eta(\alpha_{n+1}^2+\alpha_{n+1}\alpha_{n}-\beta_{n+2})-2(\gamma_n+2)\alpha_{n+1}-5\alpha_{n}\right]\right.\\ &\left.\qquad\qquad\quad+\beta_{n}\left[-\eta(\alpha_{n-1}^2+\alpha_{n}\alpha_{n-1}-\beta_{n-1})+2(\gamma_n-1)\alpha_{n-1}-\alpha_{n}\right]+\alpha_{n}\left(\alpha_n^2+1\right)\right\}\\ &+\eta(\alpha_{n}^2-\beta_{n+1}-\beta_{n}+1)-2\gamma_n\alpha_{n},\qquad n\geq2.
	\end{aligned}
\end{equation*}
Combining it with \eqref{pnalbt-1} to eliminate $\textbf{p}(n)$, and replacing $n$ by $n-1$ in the obtained equation, we arrive at \eqref{diff-bt}, valid for $n\geq3$.
\end{proof}

\begin{remark}
Equations \eqref{diff-bt} and \eqref{diff-al} can be rewritten as follows:
\begin{align*}
\beta_{n+1}=&F\left(\alpha_{n},\ldots,\alpha_{n-3};\beta_{n},\ldots,\beta_{n-2} \right),& \alpha_{n+1}=&G\left(\alpha_{n},\ldots, \alpha_{n-3};\beta_{n+1},\ldots,\beta_{n-2}\right),
\end{align*}
for $n\geq3$. When $\gamma=1/2$ and $\eta=m\pi$ with $m$ denoting a nonzero integer, by using the differential relations satisfied by the monic orthogonal polynomials, two coupled difference equations of the form
\begin{align*}
\beta_{n+1}=&\widetilde{F}\left( \alpha_{n},\ldots,\alpha_{n-3};\beta_{n},\ldots, \beta_{n-3}\right),& \alpha_{n+1}=&\widetilde{G}\left(\alpha_{n},\ldots, \alpha_{n-3}; \beta_{n+1},\ldots, \beta_{n-3}\right),
\end{align*}
for $n\geq4$, were deduced in  \cite[Theorem 3.3]{MC05-2}, both of which are one order higher with respect to $\beta_n$ than ours.
\end{remark}

Consider the following two cases:
\begin{itemize}
\item [(i)] $\gamma=0$ and $\eta$ is a positive zero of the Bessel function $J_0$, i.e., $w_2(x)=x(1-x^2)^{-1/2}\exp(\mathit{i}\eta x)$;
\item [(ii)] $\gamma=1/2$ and $\eta=m\pi$ with $m$ denoting a nonzero integer, i.e., $w_2(x)=x\exp(\mathit{i}\eta x)$.
\end{itemize}
We use \eqref{dp}-\eqref{al} and \eqref{bt} to calculate the initial values $\{\alpha_0,\alpha_1,\alpha_2,\alpha_3\}$ and $\{ \beta_1,\beta_2,\beta_3\}$.
Specifically, given the initial values
\[P_{-1}(x)=0,\qquad\qquad P_0(x)=1,\]
we use the integral expressions \eqref{al-1} and \eqref{bt} to compute $\alpha_n$ and $\beta_n$ respectively, i.e.
\begin{gather}
\alpha_n=\frac{-i}{h_n}\int_{-1}^1P_n^2(x)x w_2(x)dx,\qquad n\geq0,\label{pf0-eq2}\\
\beta_n=\frac{h_n}{h_{n-1}},\qquad n\geq1,
\end{gather}
where $h_n$ is given by \eqref{dp} with $m=n$ and $P_n$ satisfies the recurrence relation \eqref{3r}, namely,
\begin{gather}
h_n=\int_{-1}^1 P_n^2(x)w_2(x)dx,\label{pf0-eq1}\\
P_{n+1}(x)=(x-i\alpha_{n})P_{n}(x)-\beta_{n}P_{n-1}(x),\label{pf0-eq5}
\end{gather}
for $n\geq0$. Using \eqref{pf0-eq2}-\eqref{pf0-eq5}, with the Mathematica code provided in Appendix \ref{A4}, we get   $\{\alpha_0,\alpha_1,\alpha_2,\alpha_3\}$ and $\{ \beta_1,\beta_2,\beta_3\}$. They can also be derived by hand, and the derivations for Case (i) and Case (ii) are provided in Appendix \ref{J0D-int} and Remark \ref{w2-2-IV} respectively.
Substituting them into \eqref{diff-bt}-\eqref{diff-al} for $n=3,\dots,8$ and iterating, we obtain $\{\alpha_{i},\beta_i\}_{4\leq i\leq9}$ which together with the initial conditions are presented in Tables \ref{tab1:rec-coeffs} and \ref{tab2:rec-coeffs}. Note that the expressions for $\{\textbf{p}(i)\}_{0\leq i\leq 9}$, obtained in view of $\textbf{p}(n+1)=\textbf{p}(n)-i\alpha_n$ and $\textbf{p}(0)=0$, are also included.

\begin{proposition}\label{cor0}
For the weight function $x(1-x^2)^{-1/2}\exp(\mathit{i}\eta x), x\in[-1,1]$, with $\eta$ a positive zero of $J_0$, the expressions of $\{\alpha_{n}, \beta_n, \textbf{p}(n)\}_{ 0\leq n\leq9}$ are displayed in Table \ref{tab1:rec-coeffs}.
		
\begin{table}[htbp]
\centering			 \renewcommand{\arraystretch}{1.36}
			 \setlength{\tabcolsep}{3pt}
\caption{$\{\alpha_{n}, \beta_{n},\textbf{p}(n)\}_{0\leq n\leq9}$ with $\beta_0$ arbitrary}
			 \label{tab1:rec-coeffs}
\begin{tabular}{|c|c|c|c|}
\hline
$n$ & $\alpha_n$ & $\beta_n$ &$\textbf{p}(n)$  \\
\hline
0 & $\frac{1}{\eta}$ & $\beta_0$ & $0$\\
\hline
1 & $\frac{3}{\eta}+\frac{-3\eta}{\eta^2-1}$ & $1-\frac{1}{\eta^2}$  & $-\frac{\mathit{i}}{\eta}$\\
\hline
2 & $\frac{5}{\eta}+3\eta\left(\frac{1}{\eta^2-1}-\frac{\eta^2-14}{2W_2}\right)$ & $\frac{-2W_2}{\eta^2(\eta^2-1)^2}$ & $-\mathit{i}\left(\frac{4}{\eta}-\frac{3\eta}{\eta^2-1}\right)$ \\
\hline
3 & $\frac{7}{\eta}+\eta\left(\frac{3(\eta^2-14)}{2W_2}-10\frac{ W_3}{W_4}\right)$ & $\frac{(\eta^2-1)W_4}{\eta^2W_2^2}$ & $-3\mathit{i}\left(\frac{3}{\eta}-\frac{\eta(\eta^2-14)}{2W_{2}}\right)$ \\
\hline
4 & $\frac{9}{\eta}+10\eta\left(\frac{W_3}{W_4}-\frac{W_5}{2W_6}\right)$ & $-12\frac{W_2W_6}{\eta^2W_4^2}$ & $-2\mathit{i}\left(\frac{8}{\eta}-5\eta\frac{ W_{3}}{W_{4}}\right)$\\
\hline
5 & $\frac{11}{\eta}+\eta\left(5\frac{W_5}{W_6}-21\frac{W_{8}}{W_{9}}\right)$ & $\frac{W_4W_9}{\eta^2W_6^2}$ & $-5\mathit{i}\left(\frac{5}{\eta}-\eta \frac{ W_{5}}{W_{6}}\right)$\\
\hline
6 & $\frac{13}{\eta}+21\eta\left(\frac{W_{8}}{W_{9}}-\frac{W_{11}}{2W_{12}}\right)$ & $-30\frac{W_6W_{12}}{\eta^2W_9^2}$ & $-3\mathit{i}\left(\frac{12}{\eta}-7\eta\frac{ W_{8}}{W_{9}}\right)$\\
\hline
7 & $\frac{15}{\eta}+6\eta\left(7\frac{W_{11}}{4W_{12}}-6\frac{W_{15}}{W_{16}}\right)$ & $\frac{W_9W_{16}}{\eta^2W_{12}^2}$ & $-7\mathit{i}\left(\frac{7}{\eta}-\frac{3}{2}\eta\cdot\frac{W_{11}}{W_{12}}\right)$\\
\hline
8&$\frac{17}{\eta}+36\eta\left(\frac{W_{15}}{W_{16}}-\frac{W_{19}}{2W_{20}}\right)$&$-56\frac{W_{12}W_{20}}{\eta^2W_{16}^2}$& $-4\mathit{i}\left(\frac{16}{\eta}-9\eta\frac{ W_{15}}{W_{16}}\right)$\\
\hline
9&$\frac{19}{\eta}+\eta\left(18\frac{W_{19}}{W_{20}}-55\frac{W_{24}}{W_{25}}\right)$&$\frac{W_{16}W_{25}}{\eta^2W_{20}^2}$&$-9\mathit{i}\left(\frac{9}{\eta}-2\eta\frac{ W_{19}}{W_{20}}\right)$\\
\hline
\end{tabular}

\vspace{1mm}
\parbox{0.66\textwidth}{\footnotesize \textbf{Note.} $W_j$ stands for $W_j(\eta^2)$, with $W_j(\cdot)$ denoting a monic polynomial of degree $j$. Hence, $W_j(\eta^2)$ is a monic polynomial of degree $2j$ in $\eta$, containing only even powers of $\eta$.  The expressions of $W_j(\eta^2)$ for $ j=2,3,4,5,6,8,9,11,12,15,16$ are given in Appendix \ref{symbols}, while those for $j=19,20,24,25$ are omitted for brevity.}
\end{table}		
\end{proposition}

According to Table \ref{tab1:rec-coeffs}, we propose the following conjecture.
\begin{conjecture}
We conjecture for $\{\alpha_n\}_{n\geq4}$ and $\{\beta_n\}_{n\geq2}$ that
\begin{align*}
\alpha_{2k}=&\frac{4k+1}{\eta}+ k(2k+1)\eta\left(\frac{W_{k^2-1}(\eta^2)}{W_{k^2}(\eta^2)}-\frac{W_{k(k+1)-1}(\eta^2)}{2W_{k(k+1)}(\eta^2)}\right),\\ \alpha_{2k+1}=&\frac{4k+3}{\eta}+\eta\left(k(2k+1)\frac{W_{k(k+1)-1}(\eta^2)}{2W_{k(k+1)}(\eta^2)}-(k+1)(2k+3)\frac{W_{k(k+2)}(\eta^2)}{W_{(k+1)^2}(\eta^2)}\right),
\end{align*}
with $k\geq2$, and
\begin{align*}
\beta_{2k} =& -2k(2k-1)\frac{W_{k(k-1)}(\eta^2)\cdot W_{k(k+1)}(\eta^2)}{\eta^2 \left(W_{k^2}(\eta^2)\right)^2},& \beta_{2k+1} =& \frac{W_{k^2}(\eta^2)\cdot W_{(k+1)^2}(\eta^2)}{\eta^2 \left(W_{k(k+1)}(\eta^2)\right)^2},
\end{align*}
with $k\geq1$. Here $W_j(\cdot)$ is a monic polynomial of degree $j$ with $W_0(\eta^2):=1$ and $W_1(\eta^2):=\eta^2-1$. Moreover, for $\{\textbf{p}(n)\}_{n\geq4}$, we conjecture the following:
\begin{align*}
\textbf{p}(2k)=&-i k\left(\frac{4k}{\eta}-(2k+1)\eta\frac{W_{k^2-1}(\eta^2)}{W_{k^2}(\eta^2)}\right),\\
\textbf{p}(2k+1)=&-i(2k+1)\left(\frac{2k+1}{\eta}- k\eta\frac{W_{k(k+1)-1}(\eta^2)}{2W_{k(k+1)}(\eta^2)}\right), \end{align*}
with $k\geq2$.
\end{conjecture}

\begin{proposition}
For the weight function $x\exp(\mathit{i}\eta x), x\in[-1,1],$
where $\eta=m\pi$ with $m$ a nonzero integer, the values of $\{\alpha_{n},\beta_n, \textbf{p}(n)\}_{ 0\leq n\leq 9}$ are given in Table \ref{tab2:rec-coeffs}.	
		\begin{table}[H]
			\centering
			 \renewcommand{\arraystretch}{1.4}
			 \setlength{\tabcolsep}{3pt}
			 \caption{$\{\alpha_{n},\beta_{n},\textbf{p}(n)\}_{0\leq n \le 9}$ with $\beta_0$ arbitrary}
			 \label{tab2:rec-coeffs}
			 \begin{tabular}{|c|c|c|c|}
\hline
$n$ & $\alpha_n$ & $\beta_n$ & $\textbf{p}(n)$  \\
\hline
0 & $\frac{2}{\eta}$ & $\beta_0$ & $0$\\
\hline
1 & $\frac{4}{\eta}-\frac{4\eta}{\eta^2-2}$ & $1-\frac{2}{\eta^2}$ & $-\frac{2\mathit{i}}{\eta}$\\
\hline
2&$\frac{6}{\eta}+2\eta\left(\frac{2}{\eta^2-2}-\frac{\eta^2-24}{G_2}\right)$ & $\frac{-4G_2}{\eta^2(\eta^2-2)^2}$ & $-2\mathit{i}\left(\frac{3}{\eta}-\frac{2\eta}{\eta^2-2}\right)$ \\
\hline
3&$\frac{8}{\eta}+2\eta\left(\frac{\eta^2-24}{G_2}-6\frac{G_3}{G_4}\right)$&$\frac{(\eta^2-2)G_4}{\eta^2G_2^2}$ & $-2\mathit{i}\left(\frac{6}{\eta}-\frac{\eta(\eta^2-24)}{G_{2}}\right)$ \\
\hline	
4&$\frac{10}{\eta}+6\eta\left(2\frac{G_3}{G_4}-\frac{G_5}{G_6}\right)$ & $-16\frac{G_2G_6}{\eta^2G_4^2}$ & $-4\mathit{i}\left(\frac{5}{\eta}-3\eta\frac{ G_{3}}{G_{4}}\right)$\\
\hline
5&$\frac{12}{\eta}+6\eta\left(\frac{G_5}{G_6}-\frac{4G_8}{G_9}\right)$ & $\frac{G_4G_9}{\eta^2G_6^2}$ & $-6\mathit{i}\left(\frac{5}{\eta}-\eta\frac{ G_{5}}{G_{6}}\right)$\\
\hline
6&$\frac{14}{\eta}+12\eta\left(\frac{2G_8}{G_9}-\frac{G_{11}}{G_{12}}\right)$ & $-36\frac{G_6G_{12}}{\eta^2G_9^2}$ & $-6\mathit{i}\left(\frac{7}{\eta}-4\eta\frac{ G_{8}}{G_{9}}\right)$\\
\hline
7&$\frac{16}{\eta}+2\eta\left(\frac{6G_{11}}{G_{12}}-20\frac{G_{15}}{G_{16}}\right)$ & $\frac{G_9G_{16}}{\eta^2G_{12}^2}$ & $-4\mathit{i}\left(\frac{14}{\eta}-3\eta\frac{ G_{11}}{G_{12}}\right)$\\\hline
8&$\frac{18}{\eta}+20\eta\left(\frac{2G_{15}}{G_{16}}-\frac{G_{19}}{G_{20}}\right)$&$-64\frac{G_{12}G_{20}}{\eta^2G_{16}^2}$ & $-8\mathit{i}\left(\frac{9}{\eta}-5\eta\frac{G_{15}}{G_{16}}\right)$\\
\hline
9&$\frac{20}{\eta}+20\eta\left(\frac{G_{19}}{G_{20}}-3\frac{G_{24}}{G_{25}}\right)$&$\frac{G_{16}G_{25}}{\eta^2G_{20}^2}$ & $-10\mathit{i}\left(\frac{9}{\eta}-2\eta\frac{G_{19}}{G_{20}}\right)$\\
\hline
\end{tabular}

\vspace{1mm}
\parbox{0.64\textwidth}{\footnotesize \textbf{Note.} $G_j$ stands for $G_j(\eta^2)$, with $G_j(\cdot)$ denoting a monic polynomial of degree $j$. Hence, $G_j(\eta^2)$ is a monic polynomial of degree $2j$ in $\eta$, containing only even powers of $\eta$. The expressions of $G_j(\eta^2)$ for $j=2,3,4,5,6,8,9,11,12,15,16,$ are given in Appendix \ref{symbols}, while those for $j=19,20,24,25$ are not displayed for brevity.}
\end{table}
\end{proposition}

\begin{remark}\label{w2-2-IV}
The initial values $\{\alpha_0,\alpha_1,\alpha_2,\alpha_3\}$ and $\{\beta_1,\beta_2,\beta_3\}$ given in Table \ref{tab2:rec-coeffs} can also be calculated by hand. In fact, for $\eta=m\pi$, using Euler's formula $e^{i\eta x}=\cos (\eta x)+i \sin (\eta x)$ and through integration by parts, we obtain the following integral identities:
\begin{align*}
\int_{-1}^1 x^k e^{i\eta x}dx=&\frac{-2i(-1)^m}{\eta^k}\tau_k,\qquad 1\leq k\leq 8,
\end{align*}
where $\tau_1=1, \tau_2=2i,\tau_3=\eta^2-6,\tau_4=4i\tau_3, \tau_5=\eta^4-20 \eta^2+120,\tau_6=6i\tau_5,\tau_7=\eta ^6-42
   \eta ^4+840 \eta
   ^2-5040,\tau_8=8i\tau_7$.
With the aid of these integral formulas, using \eqref{pf0-eq2}-\eqref{pf0-eq5}, we obtain the desired values.
\end{remark}

According to Table \ref{tab2:rec-coeffs}, we formulate the following conjecture.
\begin{conjecture}
We conjecture for $\{\alpha_n\}_{n\geq4}$ and $\{\beta_n\}_{n\geq2}$ that
\begin{align*}
\alpha_{2k}=&\frac{2(2k+1)}{\eta}+k(k+1)\eta\left(\frac{2G_{k^2-1}(\eta^2)}{G_{k^2}(\eta^2)}-\frac{G_{k(k+1)-1}(\eta^2)}{G_{k(k+1)}(\eta^2)}\right),\\ \alpha_{2k+1}=&\frac{2(2k+2)}{\eta}+\eta\left(k(k+1)\frac{G_{k(k+1)-1}(\eta^2)}{G_{k(k+1)}(\eta^2)}-2(k+1)(k+2)\frac{G_{k(k+2)}(\eta^2)}{G_{(k+1)^2}(\eta^2)}\right),
\end{align*}
with $k\geq2$, and
\begin{align*}
\beta_{2k} =& -(2k)^2\frac{G_{k(k-1)}(\eta^2)\cdot G_{k(k+1)}(\eta^2)}{\eta^2 \left(G_{k^2}(\eta^2)\right)^2},& \beta_{2k+1} =& \frac{G_{k^2}(\eta^2)\cdot G_{(k+1)^2}(\eta^2)}{\eta^2 \left(G_{k(k+1)}(\eta^2)\right)^2},
\end{align*}
with $k\geq1$. Here $G_j(\cdot)$ is a monic polynomial of degree $j$, with $G_0(\eta^2):=1$ and $G_1(\eta^2):=\eta^2-2$. Moreover, $\{\textbf{p}(n)\}_{n\geq4}$ have the following symbolic forms
\begin{align*}
\textbf{p}(2k)=&-2 i k\left(\frac{2k+1}{\eta}-(k+1)\eta\frac{G_{k^2-1}(\eta^2)}{G_{k^2}(\eta^2)}\right),\\
\textbf{p}(2k+1)=&-i(k+1)\left(\frac{2(2k+1)}{\eta}-k\eta\frac{G_{k(k+1)-1}(\eta^2)}{G_{k(k+1)}(\eta^2)}\right), \end{align*}
for $k\geq2$.
\end{conjecture}

\begin{remark}
Our conjectured expressions for $\alpha_n$ and $\beta_n$ are consistent with the ones given in \cite[Conjecture 3.1]{MC05-2}, while those for $\textbf{p}(n)$ are new.
\end{remark}

\begin{remark}
When $\eta=100\pi$, with the initial values $\{\alpha_0,\alpha_1,\alpha_2,\alpha_3\}$ and $\{\beta_1,\beta_2,\beta_3\}$ given in Table \ref{tab2:rec-coeffs}, we iterate \eqref{diff-bt}-\eqref{diff-al} for $3\leq n\leq 30$. The Mathematica code is provided in Appendix \ref{A6}, and the obtained values of $\alpha_n$ and $\beta_n$ for $0\leq n \leq 29$ are the same as those listed in \cite[Table 2]{MC05-2}, which were computed using a combination of the Chebyshev method and the Stieltjes-Gautschi procedure. In addition, we find
\begin{align*}
\alpha_{30}=&1.269928185955803899,&
\beta_{30}=&-0.00428212384856920060.
\end{align*}

For $\eta=\pi$, using the Mathematica code similar to that presented in Appendix \ref{A6} (see the comments therein), we compute $\alpha_n$ and $\beta_n$ for $0\leq n \leq 29$. The obtained values are identical to those in \cite[Table 2]{MC05-2}, except that the last digit of $\alpha_{18}$ is 3 (not 2), the last digit of $\beta_8$ is 8 (not 7), and the last two digits of $\beta_{22}$ are 70 (not 69). In addition, we get
\begin{align*}
\alpha_{30}=&-0.08657434839001662,&
\beta_{30}=&0.2103637730884711.
\end{align*}

The above numerical results validate our coupled difference equations \eqref{diff-bt}-\eqref{diff-al}.
\end{remark}

\begin{remark}
Since $\eta$ is real in both cases $\gamma=0$ and $\gamma=1/2$,  we know that the initial values $\{\alpha_0,\alpha_1,\alpha_2,\alpha_3\}$ and $\{\beta_1,\beta_2,\beta_3\}$ given in Tables \ref{tab1:rec-coeffs} and \ref{tab2:rec-coeffs} are also real. Noting that the coefficients that appear in the difference equations \eqref{diff-bt}-\eqref{diff-al} are likewise real, we conclude that $\alpha_n (n\geq0)$ and $\beta_n (n\geq1)$ are real. In view of $\textbf{p}(n)=-i\sum_{j=0}^{n-1}\alpha_j$, it follows that $\textbf{p}(n)$ is purely imaginary.
\end{remark}

\section{Conclusions}
We compute the three-term recurrence coefficients (i.e. $i\alpha_n$ and $\beta_n$) of the monic orthogonal polynomials (OPs for short) associated with $w_1(x)=(1-x^2)^{\lambda-1/2}\exp(i\zeta x)$ and $w_2(x)=x(1-x^{2})^{\gamma-1/2} \exp(\mathit{i}\eta x)$, where $x\in[-1,1]$, $\lambda>-1/2$ is rational, $\zeta\in\mathbb{R}\setminus\{0\}$ is a zero of  $J_{\lambda-1}$, and $\gamma,\eta$ are chosen to guarantee the existence of the OPs. For both weights, by using the conditions $\{(S_1), (S_2), (S_2')\}$ compatible with the ladder operators satisfied by the monic OPs, we establish two coupled difference equations for $\alpha_n$ and $\beta_n$, which can be iterated in $n$ given initial values.

For $w_1(x)$, we make use of $(S_1)$ and $(S_2')$ to derive difference equations of the form $\beta_{n+1}=f(\alpha_{n-1}, \alpha_n,\beta_n)$ and $\alpha_{n+1}=g(\alpha_n,\beta_n,\beta_{n+1})$ for $n\geq1$, which are much simpler than those of the form $\beta_{n+1}=\tilde{f}(\alpha_{n-1}, \alpha_n,\beta_{n-1})$ and $\alpha_{n+1}=\tilde{g}(\alpha_{n-1},\alpha_n,\beta_{n-1},\beta_n,\beta_{n+1})$  for $n\geq2$ in previous work. We also deduce a third-order difference equation for $\textbf{p}(n)$, the sub-leading coefficient of the $n$th-degree monic OP.
After computing the first few terms of $\alpha_n,\beta_n,$ and $\textbf{p}(n)$, we conjecture their symbolic forms for general $n$. As in prior studies, these forms for $\alpha_n$ and $\beta_n$ are rational functions involving polynomials in $\zeta$ of prescribed degrees; however, our polynomials have leading coefficient 1, whereas prior ones were left undetermined.
Our conjectured forms for $\textbf{p}(n)$, to our knowledge, are new. In the degenerate case $\lambda=0$, the conjectured expressions for $\alpha_n,\beta_n$, and $\textbf{p}(n)$ differ substantially from those in the general case. In addition, when $\lambda=1$ and $\zeta$ is a particular zero of $J_0$, the numerical results of $\alpha_n$ and $\beta_n$ we obtain for
$n=1,\dots,8$ agree with earlier research to within relative errors on the order of
$10^{-9}$ or less, which confirms the validity of our equations.

For $w_2(x)$, by introducing two auxiliary quantities and employing $\{(S_1),(S_2)\}$, we deduce a pair of coupled difference equations which, as a whole, is fourth order in $\alpha_n$ and  third order in $\beta_n$. Two special cases of $(\gamma,\eta)$ have been considered in the literature: (i) $\gamma=0$ and $\eta$ is a positive zero of $J_0$; (ii) $\gamma=1/2$ and $\eta=m\pi$ with $m$ denoting a nonzero integer. For Case (ii), two coupled difference equations were also derived in previous work, both of which are one order higher in $\beta_n$ than ours. For these two cases, we compute the desired initial values, obtain the first few terms of $\alpha_n$ and $\beta_n$, and then those of $\textbf{p}(n)$ via $\textbf{p}(n)=\textbf{p}(n-1)-i\alpha_{n-1}$. Based on these expressions, we conjecture the symbolic forms of $\alpha_n,\beta_n$, and $\textbf{p}(n)$ for arbitrary $n$. These forms for Case (i) and those of $\textbf{p}(n)$ for Case (ii), to our knowledge, have not been reported. In addition, for $m=100$ and $m=1$ in Case (ii), our numerical results for $\alpha_n$ and $\beta_n$ match those in prior literature to 18 and 15 significant digits, respectively, which validates our equations.

Detailed computations of the initial values, the iterative evaluation of
$\alpha_n,\beta_n,\textbf{p}(n)$
  from the difference equations, and the numerical implementation-all carried out in Mathematica-are provided in the Appendices.

\section*{Acknowledgements}
This work was supported by National Natural Science Foundation of China under grant numbers 12101343 and 12371257, and by Shandong Provincial Natural Science Foundation with project number ZR2021QA061.

\bibliographystyle{}

\begin{thebibliography}{10}

\bibitem{BCE}
E. Basor, Y. Chen and T. Ehrhardt, Painlev\'{e} V and time-dependent Jacobi polynomials, J. Phys. A: Math. Theor. {\bf 43} (2010), 015204 (25pp).

\bibitem{BCH}
E. Basor, Y. Chen and N. Haq, Asymptotics of determinants of Hankel matrices via non-linear difference equations, J. Approx. Theory {\bf 198} (2015), 63-110.

\bibitem{CI}
Y. Chen and M. Ismail, Jacobi polynomials from compatibility conditions, Proc. Am. Math. Soc. {\bf 133} (2005), 465-472.

\bibitem{CI10}
Y. Chen and A. Its, Painlev\'{e} III and a singular linear statistics in Hermitian random matrix ensembles, I, J. Approx. Theory {\bf 162} (2010), 270-297.


\bibitem{CM12}
Y. Chen and M. McKay, Coulomb fluid, Painlev\'{e} transcendents and the information theory of MIMO systems, IEEE Trans. Inf. Theory {\bf 58} (2012), 4594-4634.

\bibitem{CYC}
R. Chen, D. Yu and J. Chen, Computation of oscillatory integrals with an exponential kernel
and Jacobi-type singularities, J. Math. Anal. Appl. {\bf 494} (2021), 124448 (11pp).

\bibitem{CJ21}
P. A. Clarkson and K. Jordaan, Generalised Airy polynomials, J. Phys. A: Math. Theor. {\bf 54} (2021), 185202 (28pp).

\bibitem{GR}
I. S. Gradshteyn and I. M. Ryzhik, Table of Integrals, Series, and Products, Seventh edition. Edited by Daniel Zwillinger and Victor Moll. Academic Press, Elsevier, 2015.

\bibitem{Ismail}
M. Ismail, Classical and Quantum Orthogonal Polynomials in One Variable, Encyclopedia of Mathematics and its Applications 98, Cambridge University Press, Cambridge, 2005.

\bibitem{KX}
H. Kang and Q. Xu, Quadrature formulae of many highly oscillatory Fourier-type integrals with algebraic or logarithmic singularities and their error analysis, Appl. Math. Comput. {\bf 442} (2023), 127758 (15pp).

\bibitem{KHM}
D. K. Kurto\v{g}lu, A. I. Has\c{c}elik and G. V. Milovanovi\'c, A method for efficient computation of integrals with oscillatory and singular integrand, Numer. Algorithms {\bf 85} (2020), 1155-1173.

\bibitem{Lebedev}
N. N. Lebedev, Special Functions and Their Applications, Revised English Edition Translated and Edited by Richard A. Silverman. Dover Publications, New York, 1972.


\bibitem{LyuLyu25}
S. Lyu and Y. Lyu, Ladder Operators for Laguerre-type and Jacobi-type Orthogonal Polynomials, J. Phys. A: Math. Theor. {\bf 59} (2026), 115201 (30pp).

\bibitem{MC05-1}
G. V. Milovanovi\'c and A. S. Cvetkovi\'c,  Orthogonal polynomials related to the oscillatory Chebyshev weight function, {Bull. Cl. Sci. Math. Nat. Sci. Math. 30 (2005)}, {47-60}.


\bibitem{MC05-2}
G V. Milovanovi\'c and A. S. Cvetkovi\'c, Orthogonal polynomials and Gaussian quadrature rules related to oscillatory weight functions, J. Comput. Appl. Math. 179 (2005), 263-287.


\bibitem{MCM}
G. V. Milovanovi\'{c}, A. S. Cvetkovi\'{c} and Z. M. Marjanovic, Orthogonal polynomials for the oscillatory-Gegenbauer weight, Publ. Inst. Math. (Belgr.) (N.S.) {\bf 84} (2008), 49-60.


\bibitem{MCS}
G. V. Milovanovi\'{c}, A. S. Cvetkovi\'{c} and M. P. Stani\'{c}, Orthogonal polynomials for modified Gegenbauer weight and corresponding quadratures, Appl. Math. Lett. {\bf 22} (2009), 1189-1194.


\bibitem{MSM}
G. V. Milovanovi\'{c}, M. P. Stani\'{c} and T. V. Tomovi\'{c} Mladenovi\'{c}, Gaussian type quadrature rules related to the oscillatory modification of the generalized Laguerre weight functions, J. Comput. Appl. Math. {\bf 437} (2024), 115476 (8pp).

\bibitem{MC17}
C. Min and Y. Chen, Gap probability distribution of the Jacobi unitary ensemble: an elementary treatment, from finite $n$ to double scaling, Stud. Appl. Math. {\bf 140} (2017), 202-220.

\bibitem{MC20}
C. Min and Y. Chen, Painlev\'{e} V and the Hankel determinant for a singularly perturbed Jacobi weight, Nucl. Phys. B {\bf 961} (2020), 115221 (25pp).

\bibitem{MC21}
C. Min and Y. Chen, Differential, difference and asymptotic relations for Pollaczek-Jacobi type orthogonal polynomials and their Hankel determinants, Stud. Appl. Math. {\bf 147} (2021), 390-416.


\bibitem{MF25}
C. Min and P. Fang, Generalized Airy polynomials, Hankel determinants and asymptotics, Physica D {\bf 473} (2025), 134560 (9pp).


\bibitem{MW26-1}
C. Min and X. Wu, Asymptotics of the Hankel determinant and orthogonal polynomials arising from the information theory of MIMO systems, arXiv: 2510.06739.


\bibitem{MW26-2}
C. Min and X. Wu, Orthogonal polynomials for the singularly perturbed Laguerre weight, Hankel determinants and asymptotics, arXiv: 2511.09362.


\bibitem{ML25}
X. Mu and S. Lyu, Hankel determinants for a deformed Laguerre weight with multiple variables and generalized Painlev\'{e} V equation, J. Math. Phys. {\bf 66} (2025), 073506 (21pp).


\bibitem{SC}
M. P. Stani\'{c} and A. S. Cvetkovi\'{c}, Orthogonal polynomials with respect to modified
Jacobi weight and corresponding quadrature rules
of Gaussian type, Numer. Math. Theor. Meth. Appl. {\bf 4} (2011), 478-488.


\bibitem{Watson}
G. N. Watson, A Treatise on the Theory of Bessel Functions, 2nd Edn., Cambridge University Press, Cambridge, 1922.

\bibitem{ZhouLyu26}
X. Zhou and S. Lyu, Recurrence coefficients of oscillatory Jacobi orthogonal
polynomials and difference equations, J. Phys.: Conf. Ser. {\bf 3334} (2026), 012005 (8pp).


\end{thebibliography}

{\small \begin{appendices}

\section{Derivation of the initial values $\{\alpha_0,\alpha_1,\beta_1\}$ given by \eqref{al0bt1}.}\label{J0D-int-1}
We need the following lemma.
\begin{lemma}\label{w1L}
Supposing $\zeta$ is a positive zero of the Bessel function $J_{\lambda-1}$, we have the following six integral identities:
\begin{align}
\int_{-1}^{1}(1-x^2)^{\lambda-1/2}\exp(\mathit{i}\zeta x)\, dx=&A\zeta^{-\lambda}J_{\lambda}(\zeta),\label{Bs-1-1}\\
\int_{-1}^{1}x(1-x^2)^{\lambda-1/2}\exp(\mathit{i}\zeta x)\, dx=&2i A\lambda\zeta^{-\lambda-1}J_{\lambda}(\zeta),\label{Bs-1-2}\\
\int_{-1}^{1}x^{2}(1-x^2)^{\lambda-1/2}\exp(\mathit{i}\zeta x)\, dx=&A\zeta^{-\lambda}J_{\lambda}(\zeta)\left[1-\frac{2\lambda(2\lambda+1)}{\zeta^2}\right],\label{Bs-1-3}\\
\int_{-1}^{1}x^{3}(1-x^2)^{\lambda-1/2}\exp(\mathit{i}\zeta x)\, dx=&iA\zeta^{-\lambda-1}J_{\lambda}(\zeta)\left[4\lambda+1-\frac{2\lambda(2\lambda+1)(2\lambda+2)}{\zeta^2}\right].\label{Bs-1-4}
\end{align}
where $A=A(\lambda):=\sqrt{\pi}2^{\lambda}\Gamma\left(\lambda+1/2\right)$.
\end{lemma}

\begin{proof}
We first recall the following properties of the Bessel function \cite[(5.3.6)]{Lebedev}:
\begin{align}
\frac{2\nu}{z}J_{\nu}(z)=&J_{\nu-1}(z)+J_{\nu+1}(z),\label{Bs1}\\
2J_{\nu}'(z)=&J_{\nu-1}(z)-J_{\nu+1}(z).\label{Bs2}
\end{align}
Eliminating $J_{\nu-1}(z)$ from these two formulas yields
\begin{align}\label{Bs3}
J_{\nu}'(z)=\frac{\nu}{z}J_{\nu}(z)-J_{\nu+1}(z).
\end{align}
From \eqref{Bs2} follows
\begin{align}
J_{\nu+1}'(z)=&\frac{1}{2}\left(J_{\nu}(z)-J_{\nu+2}(z)\right)\nonumber\\
=&J_{\nu}(z)-\frac{\nu+1}{z}J_{\nu+1}(z),\label{Bs4}
\end{align}
where the second identity is due to the following relation which comes from \eqref{Bs1}:
\[J_{\nu+2}(z)=\frac{2(\nu+1)}{z}J_{\nu+1}(z)-J_{\nu}(z).\]
Noting that $J_{\lambda-1}(\zeta)=0$, we get from \eqref{Bs1} that
\begin{align}\label{Bs5}
\frac{2\lambda}{\zeta}J_{\lambda}(\zeta)=J_{\lambda+1}(\zeta).
\end{align}

To continue, from \cite[eq. 10, p. 912]{GR}, we know that
\begin{align}\label{h0-1}
\int_{-1}^1(1-x^2)^{\nu-1/2}\exp(iz x)dx=A(\nu)z^{-\nu}J_{\nu}(z),\qquad \nu>-1/2,
\end{align}
with $A(\nu)=\sqrt{\pi}2^{\nu}\Gamma\left(\nu+1/2\right)$, which gives \eqref{Bs-1-1}.
Differentiating it with respect to $z$, with the help of \eqref{Bs3}, we find
\begin{align}
i\int_{-1}^{1}x(1-x^2)^{\nu-1/2}\exp(\mathit{i}z x)\, dx
=&-Az^{-\nu}J_{\nu+1}(z),\label{x1-2}
\end{align}
which, in light of \eqref{Bs5}, leads us to \eqref{Bs-1-2}.
By a subsequent differentiation of \eqref{x1-2} with respect to $z$, with the aid of \eqref{Bs4}, we obtain
\begin{align}
\int_{-1}^{1}x^2(1-x^2)^{\nu-1/2}\exp(\mathit{i}z x)\, dx
=&Az^{-\nu}\left(J_{\nu}(z)-\frac{2\nu+1}{z} J_{\nu+1}(z)\right),\label{Bs6}
\end{align}
which, in view of \eqref{Bs5}, gives us \eqref{Bs-1-3}.
Taking the derivative of \eqref{Bs6} with respect to $z$ and applying \eqref{Bs3}-\eqref{Bs4} results in
\begin{align}
i\int_{-1}^{1}x^3(1-x^2)^{\nu-1/2}\exp(\mathit{i}z x)\, dx
=& Az^{-\nu}\left[-\frac{2\nu+1}{z}J_{\nu}(z)+\left(\frac{(2\nu+1)(2\nu+2)}{z^2}-1\right)J_{\nu+1}(z)\right].\label{Bs7}
\end{align}
According to \eqref{Bs5}, we come to \eqref{Bs-1-4}.
\end{proof}

Recall \eqref{al-1} and \eqref{bt}:
\begin{gather}
\alpha_n=\frac{-i}{h_n}\int_{-1}^1xP_n^2(x)(1-x^2)^{\lambda-1/2}\exp(\mathit{i}\zeta x)dx, \qquad\qquad n\geq0,\label{al-int}\\
\beta_n=\frac{h_n}{h_{n-1}}, \qquad\qquad n\geq1,\label{bt-h}
\end{gather}
where $h_n$ is given by \eqref{dp} with $m=n$, namely
\begin{gather}\label{h-int}
h_n=\int_{-1}^1P_n^2(x)(1-x^2)^{\lambda-1/2}\exp(\mathit{i}\zeta x)dx,\qquad\qquad n\geq0,
\end{gather}
and $P_n$ can be computed using the three-term recurrence relation, i.e.
\begin{gather}
P_{n+1}(x)=(x-i\alpha_{n})P_{n}(x)-\beta_{n}P_{n-1}(x),\qquad\qquad n\geq0,\label{P-3r}
\end{gather}
with $P_{-1}(x)=0$ and $P_0(x)=1$.

Taking $n=0$ in \eqref{h-int}, \eqref{al-int}, \eqref{P-3r}, and applying \eqref{Bs-1-1}-\eqref{Bs-1-2}, we have
\begin{align*}
h_0=&\int_{-1}^1(1-x^2)^{\lambda-1/2}\exp(\mathit{i}\zeta x)dx=A\zeta^{-\lambda}J_{\lambda}(\zeta),\\
\alpha_0=&\frac{-i}{h_0}\int_{-1}^1x(1-x^2)^{\lambda-1/2}\exp(\mathit{i}\zeta x)dx=\frac{2\lambda}{\zeta},\\
P_1(x)=&(x-i\alpha_{0})P_{0}(x)-\beta_{0}P_{-1}(x)=x-\frac{2 i\lambda}{\zeta}.
\end{align*}
Setting $n=1$ in \eqref{h-int} and \eqref{al-int}-\eqref{bt-h}, by applying \eqref{Bs-1-1}-\eqref{Bs-1-4}, we get
\begin{align*}
h_1=&\int_{-1}^1P_1^2(x)(1-x^2)^{\lambda-1/2}\exp(\mathit{i}\zeta x)dx=A\zeta^{-\lambda}J_{\lambda}(\zeta)\left(1-\frac{2\lambda}{\zeta^2}\right),\\
\alpha_1=&\frac{-i}{h_1}\int_{-1}^1xP_1^2(x)(1-x^2)^{\lambda-1/2}\exp(\mathit{i}\zeta x)dx=\frac{\zeta^2-4\lambda(\lambda+1)}{\zeta(\zeta^2-2\lambda)},\\
\beta_1=&\frac{h_1}{h_0}=1-\frac{2\lambda}{\zeta^2}.
\end{align*}

\section{Mathematica code to compute $\{\alpha_0,\alpha_1,\beta_1\}$ given by \eqref{al0bt1}.}\label{IV1-M}
\begin{lstlisting}[language=Mathematica]
(*initialization*)
P[0, x_] = 1;
P[-1, x_] = 0;
w[x_] = (1 - x^2)^(\[Lambda] - 1/2)*Exp[I*\[Zeta] x];
$Assumptions = Element[\[Lambda], Rationals] && \[Lambda] > -1/2;

(*calculations of \[Alpha][0], \[Alpha][1], and \[Beta][1]*)
results =
 Table[P[n + 1, x_] = (x - I*\[Alpha][n])*P[n, x] - \[Beta][n]*P[n - 1, x];
  h[n] =
   FullSimplify[
    FullSimplify[Integrate[P[n, x]^2*w[x], {x, -1, 1}]] /.
      BesselJ[\[Lambda] + 1, \[Zeta]] -> (2 \[Lambda]/\[Zeta]) BesselJ[\[Lambda], \[Zeta]]];
  \[Beta][n + 1] = FullSimplify[h[n + 1]/h[n]];
  \[Alpha][n] =
   Map[Factor,
    Apart[
     FullSimplify[
      FullSimplify[-I/h[n]*
      Integrate[x*P[n, x]^2*w[x], {x, -1, 1}]] /. {BesselJ[2 + \[Lambda], \[Zeta]] -> 2 (\[Lambda] + 1)/\[Zeta]*BesselJ[\[Lambda] + 1, \[Zeta]] - BesselJ[\[Lambda], \[Zeta]]} /. {BesselJ[\[Lambda] + 1, \[Zeta]] -> (2 \[Lambda]/\[Zeta]) BesselJ[\[Lambda], \[Zeta]]}], \[Zeta]]]; {n, \[Alpha][n], \[Beta][n]}, {n, 0, 1}]
\end{lstlisting}

\section{Explicit forms of $\{X_j, Y_j, W_j,G_j\}$ from Tables \ref{Table-1}-\ref{tab2:rec-coeffs}.}\label{symbols}
(1) Explicit forms of $X_{j}(\zeta^2)$ from Table \ref{Table-1} for $j=2,3,4,5,6,8,9,11,12$.
\begin{align*}
X_{2}:=&\zeta ^4-\zeta ^2 \lambda  (2 \lambda +5)+4 \lambda ^2,\\
2X_{3}:=&2 \zeta ^6-\zeta ^4 \left(2 \lambda ^2+11 \lambda +3\right)+2 \zeta ^2 \lambda  \left(2 \lambda ^2+17 \lambda +12\right)-8 \lambda ^2 \left(2 \lambda ^3+7 \lambda ^2+11 \lambda +6\right),\\
4X_{4}:=&4 \zeta ^8-\zeta ^6 \left(4 \lambda  (\lambda +9)+21\right)+12 \zeta ^4 \lambda  \left(4 \lambda  (\lambda +5)+15\right)-16 \zeta ^2 \lambda ^2 \lambda_1 \left(4 \lambda  (\lambda +4)+27\right)+192 \lambda ^3 \lambda_1,\\
8X_{5}:=&8 \zeta ^{10}-5 \zeta ^8 \left(16 \lambda ^2+54 \lambda +33\right)+2 \zeta ^6 \left(68 \lambda ^4+548 \lambda ^3+1233 \lambda ^2+774 \lambda +36\right)\\
&-16 \zeta ^4 \lambda  \lambda_1 \left(4 \lambda ^4+42 \lambda ^3+188 \lambda ^2+279 \lambda +54\right)+192 \zeta ^2 \lambda ^2 \lambda_1 \lambda_2 (14 \lambda +9)\\
&-384 \lambda ^3 \lambda_1^2 \left(2 \lambda ^3+7 \lambda ^2+12 \lambda +12\right),\\
8X_{6}:=&8 \zeta ^{12}-24 \zeta ^{10} \left(2 \lambda ^2+9 \lambda +6\right)+3 \zeta ^8 \left(24 \lambda ^4+228 \lambda ^3+674 \lambda ^2+527 \lambda +66\right)\\
&-4 \zeta ^6 \lambda  \lambda_1 \left(8 \lambda ^4+100 \lambda ^3+590 \lambda ^2+1431 \lambda +612\right)+144 \zeta ^4 \lambda ^2 \lambda_1 (2 \lambda +9) \left(2 \lambda ^2+11 \lambda +8\right)\\
&-192 \zeta ^2 \lambda ^3 \lambda_1^2 (2 \lambda +7) \left(2 \lambda ^2+5 \lambda +12\right)+4608 \lambda ^4 \lambda_1^2,\\
16X_{8}:=&16 \zeta ^{16}-4 \zeta ^{14} \left(12 \lambda ^2+88 \lambda +81\right)+\zeta ^{12} \left(48 \lambda ^4+832 \lambda ^3+4496 \lambda ^2+6736 \lambda +2889\right)\\
&-\zeta ^{10} \left(16 \lambda ^6+1632 \lambda ^5+17128 \lambda ^4+72456 \lambda ^3+128917 \lambda ^2+94824 \lambda +23076\right)\\
&+8 \zeta ^8 \lambda_1 \left(400 \lambda ^6+5568 \lambda ^5+29576 \lambda ^4+76800 \lambda ^3+91173 \lambda ^2+37476 \lambda +540\right)\\
&-64 \zeta ^6 \lambda  \lambda_1 \left(48 \lambda ^7+824 \lambda ^6+5924 \lambda ^5+22426 \lambda ^4+47488 \lambda ^3+51882 \lambda ^2+22527 \lambda +1080\right)\\
&+64 \zeta ^4 \lambda ^2 \lambda_1^2 \left(16 \lambda ^7+288 \lambda ^6+2344 \lambda ^5+11520 \lambda ^4+34117 \lambda ^3+57726 \lambda ^2+44892 \lambda +6480\right)\\
&+3072 \zeta ^2 \lambda ^3 \lambda_1^2 \left(4 \lambda ^5-16 \lambda ^4-253 \lambda ^3-842 \lambda ^2-1044 \lambda -360\right)\\
&+6144 \lambda ^4 \lambda_1^3 \left(4 \lambda ^5+32 \lambda ^4+107 \lambda ^3+202 \lambda ^2+258 \lambda +180\right),\\
16X_{9}:=&16 \zeta ^{18}-12 \zeta ^{16} \left(4 \lambda ^2+36 \lambda +41\right)+3 \zeta ^{14} \left(16 \lambda ^4+480 \lambda ^3+3000 \lambda ^2+5272 \lambda +2601\right)\\
&-\zeta ^{12} \left(16 \lambda ^6+3504 \lambda ^5+40344 \lambda ^4+176296 \lambda ^3+324993 \lambda ^2+246267 \lambda +61254\right)\\
&+12 \zeta ^{10} \lambda_1 \left(496 \lambda ^6+7456 \lambda ^5+43448 \lambda ^4+122600 \lambda ^3+157107 \lambda ^2+70500 \lambda +2880\right)\\
&-24 \zeta ^8 \lambda  \lambda_1 \left(208 \lambda ^7+3872 \lambda ^6+30552 \lambda ^5+129520 \lambda ^4+310013 \lambda ^3+384306 \lambda ^2+198840 \lambda +23400\right)\\
&+96 \zeta ^6 \lambda ^2 \lambda_1^2 \left(16 \lambda ^7+320 \lambda ^6+2904 \lambda ^5+16144 \lambda ^4+57113 \lambda ^3+121224 \lambda ^2+123780 \lambda +36000\right)\\
&-4608 \zeta ^4 \lambda ^3 \lambda_1^2 \left(8 \lambda ^5+180 \lambda ^4+1214 \lambda ^3+3687 \lambda ^2+4880 \lambda +2100\right)\\
&+9216 \zeta ^2 \lambda ^4 \lambda_1^3 \left(8 \lambda ^5+84 \lambda ^4+374 \lambda ^3+939 \lambda ^2+1526 \lambda +1200\right)-1105920 \lambda ^5 \lambda_1^3 \lambda_2,\\
96X_{11}:=&96 \zeta ^{22}-8 \zeta ^{20} \left(224 \lambda ^2+1238 \lambda +1431\right)+2 \zeta ^{18} \left(3872 \lambda ^4+47656 \lambda ^3+191204 \lambda ^2+290430 \lambda +140265\right)\nonumber\\
&-3 \zeta ^{16} \left(4992 \lambda ^6+94192 \lambda ^5+686048 \lambda ^4+2354200 \lambda ^3+3897744 \lambda ^2+2916651 \lambda +775944\right)\nonumber\\
&+\zeta ^{14} \left(14944 \lambda ^8+374896 \lambda ^7+3958416 \lambda ^6+22444552 \lambda ^5+71298694 \lambda ^4+125235495 \lambda ^3\right.\nonumber\\
&\left.\quad+116216268 \lambda ^2+51090228 \lambda +7671024\right)-4 \zeta ^{12} \lambda_1 \left(1888 \lambda ^9+56560 \lambda ^8+785616 \lambda ^7+6327496 \lambda ^6\right.\\
&\left.\quad+31195606 \lambda ^5+91434975 \lambda ^4+151352256 \lambda ^3+132498360 \lambda ^2+54883440 \lambda +8106480\right)\nonumber\\
&+48 \zeta ^{10} \lambda_1 \left(32 \lambda ^{11}+1136 \lambda ^{10}+31328 \lambda ^9+461256 \lambda ^8+3759082 \lambda ^7+18388343 \lambda ^6+54996991 \lambda ^5\right.\nonumber\\
&\left.\quad+98647748 \lambda ^4+101241360 \lambda ^3+53841960 \lambda ^2+11404800 \lambda +64800\right)\nonumber\\
&-384 \zeta ^8 \lambda  \lambda_1^2 \lambda_2 \left(2544 \lambda ^8+55872 \lambda ^7+515912 \lambda ^6+2603872 \lambda ^5+7871663 \lambda ^4+14013300 \lambda ^3\right.\nonumber\\
&\left.\quad+12940920 \lambda ^2+4608900 \lambda +81000\right)+1536 \zeta ^6 \lambda ^2 \lambda_1^2 \lambda_2\left(464 \lambda ^9+11168 \lambda ^8+117832 \lambda ^7\right. \nonumber\\
&\left.\quad+717296 \lambda ^6+2747717 \lambda ^5+6842338 \lambda ^4+10811199 \lambda ^3+9650160 \lambda ^2+3677400 \lambda +162000\right)\nonumber\\
&-6144 \zeta ^4 \lambda ^3 \lambda_1^3\lambda_2 \lambda_3 \left(32 \lambda ^8+656 \lambda ^7+6128 \lambda ^6+34712 \lambda ^5+136010 \lambda ^4+346241 \lambda ^3\right.\nonumber\\
&\left.\quad+553134 \lambda ^2+432900 \lambda +54000\right)-737280 \zeta ^2 \lambda ^4 \lambda_1^3 \lambda_2^2 \lambda_3 \left(20 \lambda ^4-44 \lambda ^3-427 \lambda ^2-1128 \lambda -450\right)\nonumber\\
&-2949120 \lambda ^5 \lambda_1^4 \lambda_2^2 \lambda_3 \left(4 \lambda ^4+20 \lambda ^3+49 \lambda ^2+60 \lambda +90\right),\\
64X_{12}:=&64 \zeta ^{24}-384 \zeta ^{22} \left(2 \lambda ^2+13 \lambda +16\right)+8 \zeta ^{20} \left(368 \lambda ^4+4976 \lambda ^3+22360 \lambda ^2+36940 \lambda +19143\right)\nonumber\\
&-\zeta ^{18} \left(5376 \lambda ^6+109440 \lambda ^5+876176 \lambda ^4+3360832 \lambda ^3+6190760 \lambda ^2+5159544 \lambda +1548729\right)\nonumber\\
&+9 \zeta ^{16} \left(576 \lambda ^8+15488 \lambda ^7+181968 \lambda ^6+1172448 \lambda ^5+4294188 \lambda ^4+8734584 \lambda ^3+9491103 \lambda ^2\right.\nonumber\\
&\left.\quad+5030386 \lambda +987096\right)-16 \zeta ^{14}\lambda_1 \left(160 \lambda ^9+5136 \lambda ^8+86768 \lambda ^7+861160 \lambda ^6+5075178 \lambda ^5\right.\nonumber\\
&\left.\quad+17450153 \lambda ^4+33565668 \lambda ^3+33998895 \lambda ^2+16001100 \lambda +2513430\right)\nonumber\\
&+8 \zeta ^{12} \lambda_1 \left(64 \lambda ^{11}+2432 \lambda ^{10}+106064 \lambda ^9+1940960 \lambda ^8+17987468 \lambda ^7+96445832 \lambda ^6\right.\nonumber\\
&\left.\quad+312736999 \lambda ^5+606710106 \lambda ^4+672366420 \lambda ^3+384845040 \lambda ^2+88393680 \lambda +1425600\right)\nonumber\\
&-1152 \zeta ^{10} \lambda  \lambda_1^2 \lambda_2 \left(544 \lambda ^8+13024 \lambda ^7+132512 \lambda ^6+743728 \lambda ^5+2496426 \lambda ^4+4947670 \lambda ^3\right.\nonumber\\
&\left.\quad+5112915 \lambda ^2+2070900 \lambda +99900\right)+768 \zeta ^8 \lambda ^2 \lambda_1^2 \lambda_2 \left(512 \lambda ^9+13536 \lambda ^8+158224 \lambda ^7+1075280 \lambda ^6\right.\nonumber\\
&\left.\quad+4653528 \lambda ^5+13190734 \lambda ^4+23875281 \lambda ^3+24616020 \lambda ^2+11348100 \lambda +1215000\right)\nonumber\\
&-3072 \zeta ^6 \lambda ^3 \lambda_1^3 \lambda_2 \left(32 \lambda ^9+848 \lambda ^8+10304 \lambda ^7+76280 \lambda ^6+388442 \lambda ^5+1407701 \lambda ^4+3563445 \lambda ^3\right.\nonumber\\
&\left.\quad+5860494 \lambda ^2+5008500 \lambda +1242000\right)\\
&+1105920 \zeta ^4 \lambda ^4 \lambda_1^3 \lambda_2^2 \left(104 \lambda ^4+1044 \lambda ^3+4054 \lambda ^2+7299 \lambda +3600\right)\nonumber\\
&-1474560 \zeta ^2 \lambda ^5 \lambda_1^4 \lambda_2^2 (2 \lambda +9) \left(4 \lambda ^4+28 \lambda ^3+95 \lambda ^2+161 \lambda +270\right)
+530841600 \lambda ^6 \lambda_1^4 \lambda_2^2.
\end{align*}
Here $\lambda_i:=\lambda+i$ for $i=1,2,3$.\vspace{3mm}
\\
(2) Explicit forms of $Y_{j}(\zeta^2)$ from Table \ref{Table-1-0} for $j=2,3,4,5,6,8,9,11,12,15,16$.
\begin{align*}
4Y_{2}:=&4 \zeta ^4-72 \zeta ^2+99,\\
16Y_{3}:=&16 \zeta ^6-300 \zeta ^4+2187 \zeta ^2-17622,\\
16Y_{4}:=&16 \zeta ^8-492 \zeta ^6+7803 \zeta ^4-61254 \zeta ^2+34560,\\
96Y_{5}:=&96 \zeta ^{10}-11944 \zeta ^8+291822 \zeta ^6-2328885 \zeta ^4+6414408 \zeta ^2-26226720,\\
64Y_{6}:=&64 \zeta ^{12}-6144 \zeta ^{10}+153144 \zeta ^8-1548729 \zeta ^6+8883864 \zeta ^4-40214880 \zeta ^2+11404800,\\
768Y_{8}:=&768 \zeta ^{16}-62784 \zeta ^{14}+2295136 \zeta ^{12}-83342508 \zeta ^{10}+2029268835 \zeta ^8-23810490144 \zeta ^6\\
&+113854569840 \zeta ^4-166228536960 \zeta ^2+405791078400,\\
256Y_{9}:=&256 \zeta ^{18}-26496 \zeta ^{16}+1540080 \zeta ^{14}-68175864 \zeta ^{12}+1655096733 \zeta ^{10}-19800590616 \zeta ^8\\
&+111180042720 \zeta ^6-320921429760 \zeta ^4+944392089600 \zeta ^2-151538688000,\\
12288Y_{11}:=&12288\zeta^{22}-4841984 \zeta ^{20}+498115584 \zeta ^{18}-22487382864 \zeta ^{16}+544353295644 \zeta ^{14}\nonumber\\
&-8913114152811 \zeta ^{12}+129651131201400 \zeta ^{10}-1583714648570400 \zeta ^8+11730625253222400 \zeta ^6\nonumber\\
&-38412263198822400 \zeta ^4+34441015482777600 \zeta ^2-55575779180544000,\\
256Y_{12}:=&256\zeta ^{24}-76800 \zeta ^{22}+7503840 \zeta ^{20}-351447300 \zeta ^{18}+9849042720 \zeta ^{16}-211834514439 \zeta ^{14}\nonumber\\
&+3852242371395 \zeta ^{12}-50175327507750 \zeta ^{10}+381546799239600 \zeta ^8-1420808577840000 \zeta ^6\nonumber\\
&+2410899503270400 \zeta ^4-4958458219008000 \zeta ^2+487223009280000,\\
16384Y_{15}:=&16384\zeta ^{30}-3735552 \zeta ^{28}+399025152 \zeta ^{26}-37752509952 \zeta ^{24}+3191338360512 \zeta ^{22}\nonumber\\
&-180668317734288 \zeta ^{20}+6264700943957640 \zeta ^{18}-134725503695291055 \zeta ^{16}\nonumber\\
&+1865006354944081260 \zeta ^{14}-18155774036247088680 \zeta ^{12}+149392043600809982400 \zeta ^{10}\nonumber\\
&-1120981072745389824000 \zeta ^8+5857439245474692096000 \zeta ^6-14118903261833158656000 \zeta ^4\nonumber\\
&+8447643101325754368000 \zeta ^2-9665272382965678080000,\\
16384Y_{16}:=&16384 \zeta ^{32}-4300800 \zeta ^{30}+627102720 \zeta ^{28}-76294333440 \zeta ^{26}+6704550057600 \zeta ^{24}\nonumber\\
&-371834537851296 \zeta ^{22}+12845996840802960 \zeta ^{20}-285997297566804075 \zeta ^{18}\nonumber\\
&+4385348005754977200 \zeta ^{16}-52253512139882473200 \zeta ^{14}+541564860087827107200 \zeta ^{12}\nonumber\\
&-4545805285706008320000 \zeta ^{10}+24640777114228684800000 \zeta ^8-66600404891865169920000 \zeta ^6\nonumber\\
&+72962121079040901120000 \zeta ^4-110368487460632002560000 \zeta ^2+7033911643629158400000.
\end{align*}
(3) Explicit forms of $W_j(\eta^2)$ from Table \ref{tab1:rec-coeffs} for $j=2,3,4,5,6,8,9,11,12,15,16$.
\begin{align*}
2W_2:=	&2\eta^4-15\eta^2+4,\\
2W_3:=	 &2\eta^6-8\eta^4+11\eta^2-168,\\
4W_4:=	 &4\eta^8-37\eta^6+295\eta^4-1248\eta^2+144,\\
8W_{5}:=&8\eta^{10}-773\eta^8+6761\eta^6-12522\eta^4-3072\eta^2-58176,\\
8W_{6}:=&8\eta^{12}-440\eta^{10}+3905\eta^8-11727\eta^6+32340\eta^4-110880\eta^2+6912,\\
16W_{8}:=&16\eta^{16}-436\eta^{14}+4129\eta^{12}-242046\eta^{10}+3296421\eta^8-14682816\eta^6+16861824\eta^4\\
&+8259840\eta^2+26542080,\\
16W_{9}:=&16\eta^{18}-636 \eta^{16}+21099\eta^{14}-723251\eta^{12}+8285121\eta^{10}-36265401\eta^8+56679264\eta^6\\
&-61715520\eta^4+216483840\eta^2-8294400,\\
16W_{11}:=&16\eta^{22}-5236\eta^{20}+268879\eta^{18}-4500486 \eta^{16}+26032641\eta^{14}-238158426\eta^{12}\\
&+4432046004\eta^{10}-34557365880\eta^8+105003993600\eta^6-85274726400\eta^4\\
&-41286205440\eta^2-56335564800,\\
64W_{12}:=&64\eta^{24}-12576\eta^{22}+603648\eta^{20}-10953887\eta^{18}+116382591\eta^{16}-2014897041\eta^{14}\\
&+27370556433\eta^{12}-179369958168\eta^{10}+520642779840\eta^8-515759892480\eta^6\\
&+190938746880\eta^4-932224204800\eta^2+23887872000,\\
256W_{15}:=&256\eta^{30}-24768\eta^{28}+984608\eta^{26}-99913860\eta^{24}+7272545505\eta^{22}-229238078161\eta^{20}\\
&+3336622422557\eta^{18}-22813872410709\eta^{16}+71457280019892\eta^{14}-369746460729216\eta^{12}\\
&+4741162121583360\eta^{10}-27852122634854400\eta^8+65627975905689600\eta^6\\
&-40565849142067200\eta^4-17279072442777600\eta^2-12917032353792000,\\
256W_{16}:=&256\eta^{32}-30592\eta^{30}+2610288 \eta^{28}-275214104\eta^{26}+16234800485 \eta^{24}-466147159164\eta^{22}\\
&+6682602532554 \eta^{20}-50829324276348\eta^{18}+308250455440113 \eta^{16}\\
&-3063456569159424\eta^{14}+26975907910160640 \eta^{12}-127620988169748480 \eta^{10}\\
&+279555488146022400\eta^8-196728971172249600\eta^6+7551166552473600\eta^4\\
&-185311791415296000 \eta^2+3371056496640000.
\end{align*}
(4) Explicit forms of $G_j(\eta^2)$ from Table \ref{tab2:rec-coeffs} for $j=2,3,4,5,6,8,9,11,12,15,16$.
\begin{align*}
G_{2}:=&\eta^4-13\eta^2+6,\\
G_{3}:=&\eta^6-7\eta^4+12\eta^2-360,\\
G_{4}:=&\eta^8-14\eta^6+165\eta^4-1224\eta^2+216,\\
G_{5}:=&\eta^{10}-133\eta^8+1929\eta^6-5940\eta^4-4320\eta^2-64800,\\	 G_{6}:=&\eta^{12}-77\eta^{10}+1101\eta^8-4932\eta^6+17361\eta^4-113400\eta^2+9720,\\	 2G_{8}:=&2\eta^{16}-78\eta^{14}+1089 \eta^{12}-70785\eta^{10}+1578285\eta^8-11659950 \eta^6+21402225 \eta^4\\
&+22963500\eta^2+76545000,\\
G_{9}:=&\eta^{18}-54 \eta^{16}+2286\eta^{14}-105966\eta^{12}+1927125\eta^{10}-13601520 \eta^8+31122225 \eta^6\\
&-30362850\eta^4+273375000\eta^2-13122000,\\
2G_{11}:=&2\eta^{22}-829\eta^{20}+58275\eta^{18}-1433430\eta^{16}+13530150\eta^{14}-117734175\eta^{12}\\
&+3083335875\eta^{10}-41419987875\eta^8+211076026875\eta^6-264998790000\eta^4\\
&-273265650000\eta^2-385786800000,\\
G_{12}:=&\eta^{24}-252\eta^{22}+16437\eta^{20}-425835\eta^{18}+6134265\eta^{16}-120820950\eta^{14}+2458605150\eta^{12}\\
&-26341594875\eta^{10}+124060399875\eta^8-176784778125\eta^6-8439086250\eta^4\\
&-699238575000\eta^2+20667150000,\\
G_{15}:=&\eta^{30}-126\eta^{28}+6654\eta^{26}-749439\eta^{24}+71087220\eta^{22}-3142068435\eta^{20}\\
&+67831410825\eta^{18}-727590460275\eta^{16}+3449814897375\eta^{14}-12254173249500\eta^{12}\\
&+266877060778125\eta^{10}-2987052554075625\eta^8+12107364149475000\eta^6\\
&-11246939027100000\eta^4-10207918728000000\eta^2-7655939046000000,\\
G_{16}:=&\eta^{32}-152\eta^{30}+15804\eta^{28}-2021700\eta^{26}+157758030\eta^{24}-6308511300\eta^{22}\\
&+132164631450\eta^{20}-1493816296500\eta^{18}+10920101842125\eta^{16}-117706861125000\eta^{14}\\
&+1646298202263750\eta^{12}-13233804005362500\eta^{10}+47682964163165625\eta^8\\
&-47494695181500000\eta^6-22188045568500000\eta^4-93572588340000000\eta^2\\
&+1822842630000000.
\end{align*}

\section{When $\lambda=1$ and $\zeta=8.653727912911012$, Mathematica code to compute $\alpha_n$ and $\beta_n$ given in Tables \ref{tab:comparison-al} and \ref{tab:comparison-bt}.}\label{table-value}
\begin{lstlisting}[language=Mathematica]
\[Lambda] = 1; \[Zeta] = 8.653727912911012;

(*the initial values*)
\[Alpha][0] = (2 \[Lambda])/\[Zeta];
\[Alpha][1] = (2 (1 + \[Lambda]))/\[Zeta] - (\[Zeta] (1 +2 \[Lambda]))/(\[Zeta]^2 - 2 \[Lambda]);
\[Beta][1] = 1 - (2 \[Lambda])/\[Zeta]^2;

(*computing \[Alpha][i] and \[Beta][i] by using the two difference equations they satisfy*)
Do[\[Beta][n + 1] = \[Beta][n] (2*\[Zeta] (\[Alpha][n] + \[Alpha][n - 1])/(2 n +
           2 \[Lambda] - 1) - 3) + \[Alpha][n] (\[Alpha][n] - (2 n + 2 \[Lambda] + 1)/\[Zeta]) +1;
\[Alpha][n + 1] = -\[Alpha][n] + (2 n + 2 \[Lambda] + 1)/(2 \[Zeta]) (3 + (\[Beta][n] - 1)/\[Beta][n + 1] + \[Alpha][n]/\[Beta][n + 1] ((
          2 n + 2 \[Lambda] - 1 )/\[Zeta] - \[Alpha][n]));, {n,1,14}];

(*output \[Alpha][n] and \[Beta][n] in tabular form for n=0,1,...,14*)
Grid[Prepend[
     Table[{n, ScientificForm[\[Alpha][n], 16], ScientificForm[\[Beta][n], 16]}, {n, 0, 14}], {"n", "\[Alpha][n]", "\[Beta][n]"}], Frame -> All]
\end{lstlisting}

\section{When $\gamma=0$, Mathematica code to compute $\alpha_0$ and $\{\alpha_n,\beta_n\}_{n=1,2,3}$ given in Table \ref{tab1:rec-coeffs}.}\label{A4}
\begin{lstlisting}[language=Mathematica]
(*initialization*)
P[0, x_] = 1;
P[-1, x_] = 0;
w[x_] = x*(1 - x^2)^(-1/2)*Exp[I*\[Eta] x];
$Assumptions = BesselJ[0, \[Eta]] == 0 && \[Eta] > 0;
(**when \[Gamma]=1/2 and \[Eta]=m \[Pi], with m being an arbitrary nonzero positive integer, the above two lines of code should be replaced by
      w[x_] = x*Exp[I*\[Eta] x];
      $Assumptions = {m \[Element] Integers, m > 0};
      \[Eta] = m*\[Pi];**)

(*definitions of h[0] and \[Alpha][0]*)
h[0] = Integrate[w[x], {x, -1, 1}];
\[Alpha][0] = Simplify[-I/h[0]*Integrate[x*w[x], {x, -1, 1}] /. BesselJ[2, \[Eta]] -> (2/\[Eta]) BesselJ[1, \[Eta]]];
(**when \[Gamma]=1/2, skip the substitution in the previous line**)

(*calculations of \[Alpha][n] and \[Beta][n] for n=1,2,3*)
results =
  Table[If[n > 0,
    P[n, x_] = (x - I*\[Alpha][n - 1])*P[n - 1, x] - \[Beta][n - 1]*P[n - 2, x];
    h[n] = Simplify[ Integrate[P[n, x]^2*w[x], {x, -1, 1}] /.
       BesselJ[2, \[Eta]] -> (2/\[Eta]) BesselJ[1, \[Eta]]];
  (**when \[lambda]=1/2, skip the substitution in the previous line**)
    \[Beta][n] = h[n]/h[n - 1];
    \[Alpha][n] =
     Map[Factor,
      Apart[
       Simplify[
         Simplify[-I/h[n]*Integrate[x*P[n, x]^2*w[x], {x, -1, 1}] /.
            BesselJ[2, \[Eta]] -> (2/\[Eta]) BesselJ[1, \[Eta]]] /.
          BesselJ[1, x] -> j] /. j -> BesselJ[1, x],\[Eta]]]];
   {n, \[Alpha][n], \[Beta][n]}, {n, 1, 3}];
   (**when \[Gamma]=1/2, skip the substitutions of the Bessel function in the previous lines**)

(**when \[Gamma]=1/2 and \[zeta]=m \[Pi], the following two lines of code need to be evaluated.
  Clear[\[Eta]];
  resultsFixed = results //. {m^k_*\[Pi]^k_ -> \[Eta]^k, m*\[Pi] -> \[Eta]};**)

(*output the values of \[Alpha][n] and \[Beta][n] for n=0,1,2,3 in tabular form*)
TableForm[Join[{{0, \[Alpha][0], \[Beta][0]}}, results],TableHeadings -> {None, {"n", "\[Alpha][n]", "\[Beta][n]"}}]
(**when \[Gamma]=1/2, replace "results" with "resultsfixed" in the line above**)
\end{lstlisting}

\section{When $\gamma=0$, an alternative derivation of the initial values $\{\alpha_0,\alpha_1,\alpha_2,\alpha_3\}$ and $\{\beta_1,\beta_2,\beta_3\}$ from Table \ref{tab1:rec-coeffs}.}\label{J0D-int}
We need the following lemma.
\begin{lemma}
Supposing $\eta$ is a positive zero of the Bessel function $J_0$, we have the following six integral identities:
\begin{align*}
\frac{1}{\pi}\int_{-1}^{1}x(1-x^2)^{-1/2}\exp(\mathit{i}\eta x)\, dx=&i J_1(\eta),\\
\frac{1}{\pi}\int_{-1}^{1}x^{2}(1-x^2)^{-1/2}\exp(\mathit{i}\eta x)\, dx=&-\frac{J_1(\eta)}{\eta},\\
\frac{1}{\pi}\int_{-1}^{1}x^{3}(1-x^2)^{-1/2}\exp(\mathit{i}\eta x)\, dx=&i\left(-\frac{2}{\eta^2} +1\right) J_1(\eta),\\
\frac{1}{\pi}\int_{-1}^{1}x^{4}(1-x^2)^{-1/2}\exp(\mathit{i}\eta x)\, dx=&2\left(\frac{3}{\eta^3}-\frac{1}{\eta}\right)J_1(\eta),\\
\frac{1}{\pi}\int_{-1}^{1}x^{5}(1-x^2)^{-1/2}\exp(\mathit{i}\eta x)\, dx=&i\left(\frac{24}{\eta^4}-\frac{7}{\eta^2}+1\right)J_1(\eta),\\
\frac{1}{\pi}\int_{-1}^{1}x^{6}(1-x^2)^{-1/2}\exp(\mathit{i}\eta x)\, dx=&-3\left(\frac{40}{\eta^5}-\frac{11}{\eta^3}+\frac{1}{\eta}\right)J_1(\eta),\\
\frac{1}{\pi}\int_{-1}^{1}x^{7}(1-x^2)^{-1/2}\exp(\mathit{i}\eta x)\, dx=&i \left(-\frac{720}{\eta^6}+\frac{192}{\eta^4}-\frac{15}{\eta^2}+1\right)J_1(\eta),\\
\frac{1}{\pi}\int_{-1}^{1}x^{8}(1-x^2)^{-1/2}\exp(\mathit{i}\eta x)\, dx=&4\left(\frac{1260}{\eta ^7}-\frac{330}{\eta ^5}+\frac{24}{\eta
   ^3}-\frac{1}{\eta }\right) J_1(\eta ).
\end{align*}
\end{lemma}

\begin{proof}
From \cite[eq. 10, p. 912]{GR}, we know that
\begin{align}\label{bsint}
\frac{1}{\pi}\int_{-1}^1(1-x^2)^{-1/2}\exp(i\eta x)dx=J_0(\eta).
\end{align}
Differentiating it $m$ times with respect to $\eta$ yields
\begin{align}\label{DJ0m}
\frac{i^m}{\pi}\int_{-1}^{1}x^{m}(1-x^2)^{-1/2}\exp(\mathit{i}\eta x)\, dx=\frac{d^m}{d\eta^m}J_0(\eta)=:J_0^{(m)}(\eta).
\end{align}
Recall the following properties of the Bessel function \cite[(5.3.7)]{Lebedev}:
\begin{align*}
\frac{d}{dz}(z^{\nu}J_{\nu}(z))=&z^{\nu}J_{\nu-1}(z),&\frac{d}{dz}(z^{-\nu}J_{\nu}(z))=-z^{-\nu}J_{\nu+1}(z).
\end{align*}
With $\nu=1$ in the first identity and $\nu=0$ in the second one, we have
\begin{align}
J_{1}'(z)=&-\frac{J_1(z)}{z}+J_{0}(z),\label{DJ1}\\
J_{0}'(z)=&-J_{1}(z).\label{DJ0}
\end{align}
Setting $z=\eta$ in \eqref{DJ0} yields
\[J_0'(\eta)=-J_1(\eta).\]
 Differentiating \eqref{DJ0} with respect to $z$, in light of \eqref{DJ1}, we get
\begin{align}\label{DJ0-2}
J_{0}''(z)=\frac{J_1(z)}{z}-J_{0}(z).
\end{align}
In view of $J_0(\eta)=0$, it follows that
\[J_{0}''(\eta)=\frac{J_1(\eta)}{\eta}.\]
Taking the derivative of \eqref{DJ0-2} with respect to $z$, on account of \eqref{DJ1}-\eqref{DJ0}, we produce
\begin{align}\label{DJ0-3}
J_{0}'''(z)=\left(-\frac{2}{z^2}+1\right)J_1(z)+\frac{J_0(z)}{z},
\end{align}
so that
\begin{align*}
J_{0}'''(\eta)=\left(-\frac{2}{\eta^2}+1\right)J_1(\eta).
\end{align*}
Following the same reasoning as in the determination of $J_0''(\eta)$ and $J_0'''(\eta)$, we obtain
\begin{align*}
J_0^{(4)}(\eta)=&\left(\frac{6}{\eta^3}-\frac{2}{\eta}\right)J_1(\eta),\\
J_0^{(5)}(\eta)=&\left(-\frac{24}{\eta^4}+\frac{7}{\eta^2}-1\right)J_1(\eta),\\
J_0^{(6)}(\eta)=&\left(\frac{120}{\eta^5}-\frac{33}{\eta^3}+\frac{3}{\eta}\right)J_1(\eta),\\
J_0^{(7)}(\eta)=&\left(\frac{-720}{\eta^6}+\frac{192}{\eta^4}-\frac{15}{\eta^2}+1\right)J_1(\eta),\\
J_0^{(8)}(\eta)=&\left(\frac{5040}{\eta ^7}-\frac{1320}{\eta ^5}+\frac{96}{\eta
   ^3}-\frac{4}{\eta }\right) J_1(\eta ).
\end{align*}
Substituting $J_0^{(m)}(\eta)$ for $1\leq m\leq 8$ into \eqref{DJ0m} leads us to the desired integral formulas.
\end{proof}

Using \eqref{pf0-eq2}-\eqref{pf0-eq5}, with the aid of the integral identities listed in the above lemma, by a reasoning similar to the one given below Lemma \ref{w1L}, we obtain the values of $\alpha_n$ (for $n=0,1,2,3$) and $\beta_n$ (for $n=1,2,3$).

%
%
%
%
%

\section{When $\gamma=1/2$ and $\eta=100\pi$, Mathematica code to compute $\alpha_n$ and $\beta_n$ for $0\leq n\leq30$, with $\beta_0$ arbitrary.}\label{A6}

\begin{lstlisting}[language=Mathematica]
Clear[\[Gamma], \[Alpha], \[Beta], \[Eta]];
\[Gamma] = 1/2;

(*the initial values of \[Alpha][n] for n=0,1,2,3*)
\[Alpha][0] = 2/\[Eta];
\[Alpha][1] = 4/\[Eta] - (4 \[Eta])/(-2 + \[Eta]^2);
\[Alpha][2] = 6/\[Eta] + (4 \[Eta])/(-2 + \[Eta]^2) - (2 \[Eta] (-24 +
              \[Eta]^2))/(6 - 13 \[Eta]^2 + \[Eta]^4);
\[Alpha][3] = 8/\[Eta] + (2 \[Eta] (-24 + \[Eta]^2))/(6 - 13 \[Eta]^2 + \[Eta]^4) - (12 \[Eta] (-360 + 12 \[Eta]^2 - 7 \[Eta]^4 + \[Eta]^6))/(216 - 1224 \[Eta]^2 + 165 \[Eta]^4 - 14 \[Eta]^6 + \[Eta]^8);

(*the initial values of \[Beta][n] for n=1,2,3, with \[Beta][0] arbitrary*)
\[Beta][1] = (\[Eta]^2 - 2)/\[Eta]^2;
\[Beta][2] = -((4 (6 - 13 \[Eta]^2 + \[Eta]^4))/(\[Eta]^2 (-2 + \[Eta]^2)^2));
\[Beta][3] = ((-2 + \[Eta]^2) (216 - 1224 \[Eta]^2 + 165 \[Eta]^4 - 14 \[Eta]^6 + \[Eta]^8))/(\[Eta]^2 (6 - 13 \[Eta]^2 + \[Eta]^4)^2);

(*taking \[Eta]=100*\[Pi]*)
\[Eta] = N[100 \[Pi], 100];
(**for \[Eta]=\[Pi],replacing this line by \[Eta]=N[\[Pi],100]**)

(*computing \[Alpha][i] and \[Beta][i] for i=4,...,30, by using the above initial values and the two difference equations for the recurrence coefficients*)
Do[\[Beta][n + 1] =
   -\[Alpha][-2 + n]^2 - \[Alpha][-2 + n] \[Alpha][-1 + n] + \[Alpha][-1 + n] \[Alpha][n] + \[Alpha][n]^2 -1/\[Eta] (-2 (-1 + n + \[Gamma]) \[Alpha][-2 + n] + 5 \[Alpha][-1 + n] + 2 (1 + n + \[Gamma]) \[Alpha][n]) + \[Beta][-1 + n] + 1/(\[Eta] \[Beta][n]) (\[Alpha][-1 + n] (1 + \[Alpha][-1 + n]^2) + (-2 \[Alpha][-2 + n] - \[Alpha][-1 + n] + \[Eta] \[Beta][-2 + n]) \[Beta][-1 + n]) - \[Beta][n] + ((-\[Beta][-1 + n] + \[Beta][n]) (2 + 3 (\[Alpha][-2 + n]^2 + \[Alpha][-1 + n]^2) + (-5 + 2 (n + \[Gamma]) - \[Eta] (\[Alpha][-3 + n] + 2 \[Alpha][-2 + n])) \[Beta][-2 + n] - 2 \[Beta][-1 + n] - (1 + 2 (n + \[Gamma]) - \[Eta] (2 \[Alpha][-1 + n] + \[Alpha][n])) \[Beta][n]))/(\[Eta] (-\[Alpha][-2 + n] + \[Alpha][-1 + n]) \[Beta][n]);

 \[Alpha][n + 1] =
   1/(\[Eta] \[Beta][n + 1]) (-3 \[Alpha][n]^2 - 2 \[Alpha][-1 + n] \[Alpha][n] - \[Alpha][-1 + n]^2 + 2 (\[Beta][n] - 1) + \[Beta][-1 + n] (\[Eta] (2\[Alpha][-1 + n] + \[Alpha][-2 + n]) - 2 n - 2 \[Gamma] + 3) + (\[Alpha][-1 + n] -\[Alpha][           n]) ((\[Alpha][-2 + n] - \[Alpha][n]) (\[Eta] (\[Alpha][n] + \[Alpha][-1 + n] + \[Alpha][-2 + n]) - 2 n - 2 \[Gamma]) + \[Eta] (\[Beta][n] - \[Beta][-1 + n])) + (\[Alpha][-1 + n] - \[Alpha][n])/(\[Alpha][-2 + n] - \[Alpha][-1 + n]) (\[Alpha][-1 + n]^2 + 5 \[Alpha][-2 + n]^2 + 2 + \[Beta][n] (\[Eta] (\[Alpha][n] + 2 \[Alpha][-1 + n]) - 2 n - 2 \[Gamma] - 1) -  2 \[Beta][-1 + n] + \[Beta][-2 + n] (2 n + 2 \[Gamma] - 5 - \[Eta] (2 \[Alpha][-2 + n] + \[Alpha][-3 + n])))) - 2 \[Alpha][n] + (2 n + 2 \[Gamma] + 3)/\[Eta];, {n, 3, 29}];

(*Format number with 18 significant digits; using 18 decimal places for|x|>1*)
formatNumber[x_] :=
Module[{},If[Abs[x] > 1, NumberForm[x, {Infinity, 18}], NumberForm[x, 18]]];
(**When \[Eta]=\[Pi],18 is replaced by 16**)

(*Generating formatted table for \[Alpha][n] and \[Beta][n],0 \[LessEqual] n\[LessEqual]30; \[Beta][0] keeping unformatted,others using formatNumber*)
Grid[Prepend[Table[{n, formatNumber[FullSimplify[\[Alpha][n]]],
    If[n == 0, \[Beta][n], formatNumber[FullSimplify[\[Beta][n]]]]}, {n, 0, 30}], {"n", "\[Alpha][n]", "\[Beta][n]"}], Frame -> All]
\end{lstlisting}

\end{appendices}}

\end{document}